\documentclass{aims}
\usepackage{amsmath}
  \usepackage{paralist}
  \usepackage{graphics} 
  \usepackage{epsfig} 
\usepackage{graphicx}  \usepackage{epstopdf}
 \usepackage[colorlinks=true]{hyperref}
 \usepackage[utf8]{inputenc}
  \usepackage{subcaption}
  \usepackage{stackrel}
  \usepackage{changepage}
\hypersetup{urlcolor=blue, citecolor=red}

\def\currentvolume{X}
 \def\currentissue{X}
  \def\currentyear{200X}
   \def\currentmonth{XX}
    \def\ppages{X--XX}
     \def\DOI{10.3934/xx.xx.xx.xx}

\definecolor{boxback}{gray}{0.95}
\newtheorem{theorem}{Theorem}[section]

\newtheorem{lemma}[theorem]{Lemma}

\theoremstyle{definition}
\newtheorem{definition}[theorem]{Definition}
\newtheorem{remark}{Remark}

\newcommand{\eps}{\varepsilon}
\newcommand{\pP}{\mathbf{P}}

\newcommand{\R}{\mathbb{R}}

\title[Modularity of directed networks: cycle decomposition approach] 
      {Modularity of directed networks: cycle decomposition approach}

\author[Djurdjevac Conrad, Banisch and Schütte]{}


 \email{djurdjev@mi.fu-berlin.de}
 \email{ralf.banisch@fu-berlin.de}
 \email{Christof.Schuette@fu-berlin.de}

\begin{document}
\maketitle

\centerline{\scshape Nata\v sa Djurdjevac Conrad}
\medskip
{\footnotesize
\centerline{Freie Universit\"at Berlin}
   \centerline{Department of Mathematics and Computer Science}
   \centerline{Arnimallee 6, 14195 Berlin, Germany}
   \centerline{Zuse Institute Berlin}
   \centerline{Takustrasse 7, 14195 Berlin, Germany}
} 

\medskip

\centerline{\scshape Ralf Banisch}
\medskip
{\footnotesize
 \centerline{Freie Universit\"at Berlin}
   \centerline{Department of Mathematics and Computer Science}
   \centerline{Arnimallee 6, 14195 Berlin, Germany}
}

\medskip

\centerline{\scshape Christof Schütte}
\medskip
{\footnotesize
\centerline{Freie Universit\"at Berlin}
   \centerline{Department of Mathematics and Computer Science}
   \centerline{Arnimallee 6, 14195 Berlin, Germany}
   \centerline{Zuse Institute Berlin}
   \centerline{Takustrasse 7, 14195 Berlin, Germany}
} 

\bigskip


\begin{abstract}
The problem of decomposing networks into modules (or clusters) has gained much attention in recent years, as it can account for a coarse-grained description of complex systems, often revealing functional subunits of these systems. A variety of module detection algorithms have been proposed, mostly oriented towards finding hard partitionings of undirected networks. Despite the increasing number of fuzzy clustering methods for directed networks, many of these approaches tend to neglect important directional information. In this paper, we present a novel random walk based approach for finding fuzzy partitions of directed, weighted networks, where edge directions play a crucial role in defining how well nodes in a module are interconnected. We will show that cycle decomposition of a random walk process connects the notion of network modules and information transport in a network, leading to a new, symmetric measure of node communication. Finally, we will use this measure to introduce a communication graph, for which we will show that although being undirected it inherits all necessary information about modular structures from the original network.
\end{abstract}

\section{Introduction}
In recent years complex networks have become widely recognized as a powerful abstraction of real-world systems \cite{Newman2003b}. The simple form of network representation allows an easy description of various systems, such as biological, technological, physical and social systems \cite{Newman04}. Thus, network analysis can be used as a tool to understand the structural and functional organization of real-world systems. One of the most important network substructures are \textbf{modules} (or clusters, community structures), that typically correspond to groups of nodes which are similar to each other in some sense. In the case of undirected networks, modules represent densely interconnected subgraphs having relatively few connections to the rest of the network \cite{Girvan02}. Finding modules leads to a coarse-grained description of network by smaller subunits, which in many cases have been found to correspond to functional units of real-world systems, such as protein complexes in biological networks \cite{
Chen2006}.\\
\noindent
Although a large number of different clustering algorithms \cite{Fortunato2010,Newman2003b} have been developed, most of them are oriented towards: (i) finding hard (or full, complete) partitions of networks into modules, where every node is assigned to exactly one module; (ii) analyzing undirected, unweighted networks. However, in many networks a {\em fuzzy partition}, where some nodes belong to several modules at once or to none at all, might be more natural. Moreover, many real-world networks are weighted and directed \cite{Barrat04,Newman2008,Malliaros2013} such as metabolic networks, email networks, World Wide Web, food webs etc. In such networks edge directions contain important information about the underlying system and as such should not be dismissed \cite{Fortunato2009}. So far, several approaches for finding modules in directed networks have been developed\cite{Lancichinetti11,Rosvall11,
Newman2008,Lambiotte2012}, see \cite{Malliaros2013} for an extensive overview. Among these, generalizations of modularity optimization became the most popular group of approaches \cite{Arenas2007,Newman2008,Kim2010}. We will refer to this family of methods as \textbf{generalized Newman-Girvan (gNG)} approaches. Despite their frequent application, modularity based methods are typically only sensitive to the density of links, but tend to neglect information encoded in the pattern of the directions \cite{Fortunato2009,Schaub2012}. This limitation of gNG methods (but also one-step methods and methods based on network symmetrization) is illustrated with the following example. \\

\begin{figure}[htp]
 \begin{subfigure}[b]{0.45\textwidth}
 \includegraphics[width=0.95\textwidth]{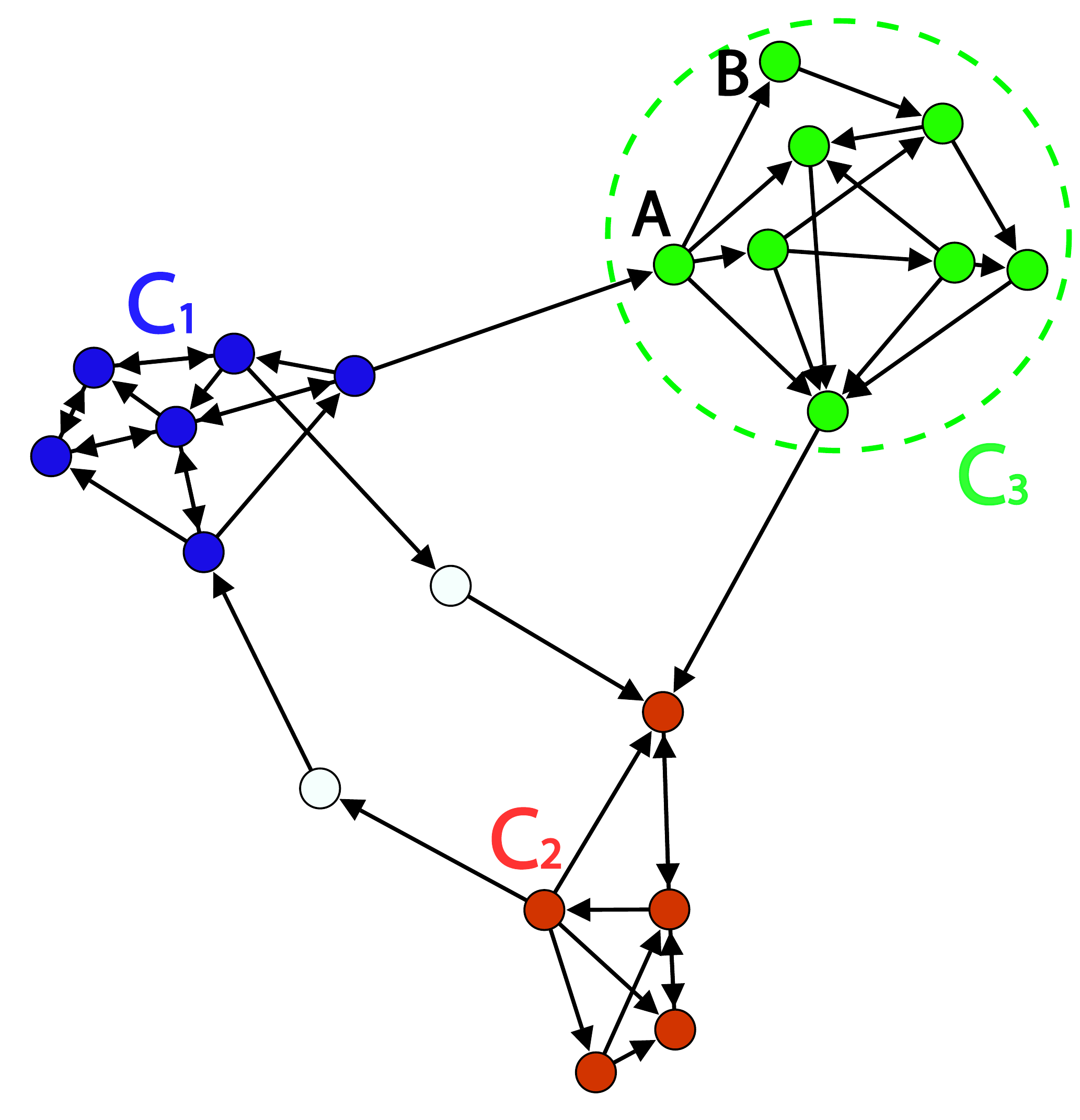}
 \caption{}
 \label{fig:marcoex_a}
 \end{subfigure}
 \begin{subfigure}[b]{0.45\textwidth}
 \includegraphics[width=0.95\textwidth]{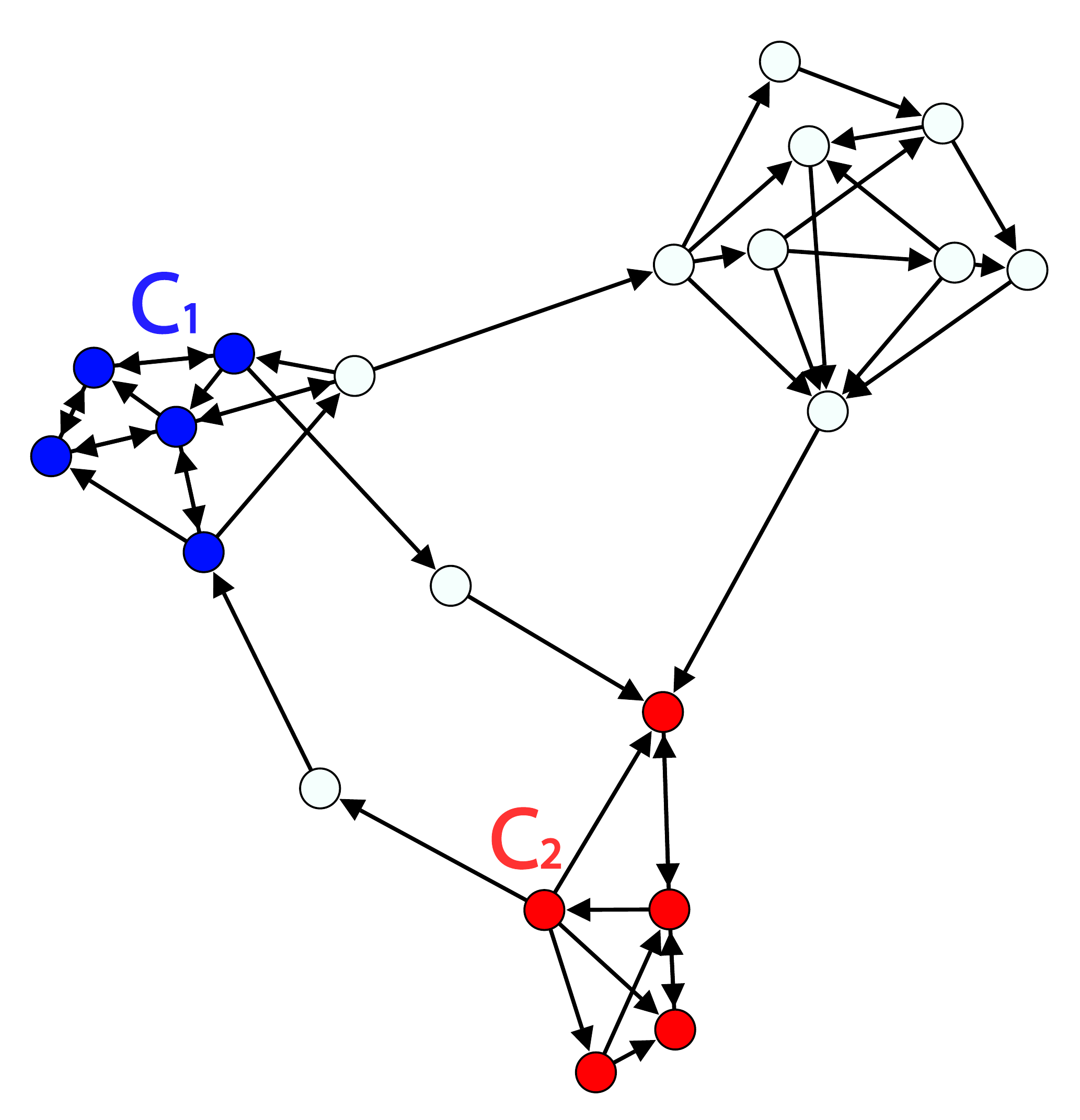}
 \caption{}
 \label{fig:marcoex_b}
 \end{subfigure}
  \caption{Illustration of a false module identification. Colors indicate modules found by (\subref{fig:marcoex_a}) gNG algorithm. (\subref{fig:marcoex_b}) our new algorithm.}
  \label{fig:marcoex_all}
\end{figure}
\noindent
Using a gNG approach, the network shown in Figure \ref{fig:marcoex_a} is partitioned into three modules $C_1$, $C_2$ and $C_3$ with high link density. However, $C_3$ is qualitatively different from $C_1$ and $C_2$: While in $C_1$ and $C_2$ edge directions are fairly symmetric, in $C_3$ they are completely asymmetric such that all nodes in $C_3$ are connected to each other by short paths in one direction and long paths in the other direction. For example, nodes $A, B\in C_3$ are connected by a directed edge $(AB)$, but in the direction from $B$ to $A$  they are connected only by long paths that pass through the whole network. In this paper, we will require nodes in a module to be close in both directions and therefore wish to reject a structure such as $C_3$ from being modular. The result of the algorithm we present here is shown in Figure \ref{fig:marcoex_b}.\\
\noindent
The main topic of this paper is on finding modules in directed, weighted networks, while explicitly respecting information encoded in the pattern of the directions. We will show that the notion of modules in directed networks is closely related to network cycles and described by dynamical properties of a random walk process, see Sections \ref{sec:flow} and \ref{sec:CycleRW}. We will make this connection precise in Section \ref{sec:cycledyn}, by using a cycle decomposition of this Markov process to introduce a novel measure of communication $I_{xy}$ between nodes. In Section \ref{sec:node_comm}, we will demonstrate that $I_{xy}$ indeed reflects information transport in the network and as such we will use $I_{xy}$ to introduce a communication graph which, although being undirected, inherits all necessary information about modular structures from the original network by construction. This will allow using some of the known fuzzy clustering approaches on the undirected communication graph, that will produce
modules of the original graph. In Section \ref{sec:Three} we will describe our new algorithmic approach and its relevant computational aspects. Finally, we will demonstrate our method on several examples in Section \ref{sec:Four} and compare it to gNG methods, indicating a possible improvement of modularity function.

\section{Theoretical background}\label{sec:Two}

\subsection{Random walks on directed networks}\label{sec:Theory}
We consider the weighted, directed network $G = (V;E)$, where $V$ is the set of nodes and $E$ the set of edges. We further assume that the network is strongly connected, i.e. there is a directed path from every node to every other node in the network. In particular, this means that the network has no sinks and sources. Now, we can define a time-discrete random walk process $(X_n)_n$ on the network with a transition matrix $P$
\begin{equation}
p_{xy} = \frac{K(x,y)}{K^+(x)},
\label{eq:RWmatrix}
\end{equation}
where $K(x,y)$ denotes the weight of an edge $(xy)\in E$ and $K^+(x) = \sum_y K(x,y)$ the weighted out-degree of a node $x$. Since the networks we consider are strongly connected, the introduced random walk process is ergodic and has a unique stationary distribution $\pi$ with $\pi^TP=\pi^T$. However, this process is not time-reversible and the detailed balance condition
\[
\pi_x p_{xy} = \pi_y p_{yx}, \quad x,y\in V
\]
doesn't hold. For this reason, most of the existing spectral methods fail to deal with finding modules in directed networks, as they are restricted to analysis of reversible processes.

\subsection{Cycle decomposition of flows}\label{sec:flow}
In this section we will briefly introduce the theory of cycle decompositions of Markov processes on finite state spaces as developed in \cite{Kalpazidou2006} and \cite{Qian2004}. We will use this decomposition to describe the random walk process defined by (\ref{eq:RWmatrix}) using cycles.\\
Let us introduce a \textbf{probability flow} $F: E\rightarrow \mathbb{R}$, such that $F_{xy} = \pi_xp_{xy}$ for every edge $(xy)\in E$. Then, the equation $\pi^T P = \pi^T$ of stationarity for $\pi$ directly yields conservation of the flow at every node $x$
\begin{equation}
\sum_y F_{yx} = \sum_y F_{xy}.
\label{eq:consFlow}
\end{equation}
Now we will generalize the notion of edge flows to cycle flows, where we define cycles in the following way
\begin{definition}
An $n$-cycle (or $n$-loop) on $G$ is defined as an ordered sequence \footnote{More precisely, cycles are equivalence classes of ordered sequences up to cycle permutations. In this note we do not distinguish between cycles and their representatives.} of $n$ connected nodes $\gamma = (x_1,x_2,\ldots, x_n)$, whose length we denote by $|\gamma| = n$. Cycles that do not have self-intersections are called \textbf{simple cycles} and we denote with $\mathcal{C}$ the collection of simple cycles in $G$.
\end{definition}
\noindent
In Figure \ref{fig:CycleEx} are shown examples of two simple cycles $\alpha$ and $\beta$, where $|\alpha| = 4$ and $|\beta| = 2$.\\

\begin{figure}[htp]
\begin{center}
 \includegraphics[width=\textwidth]{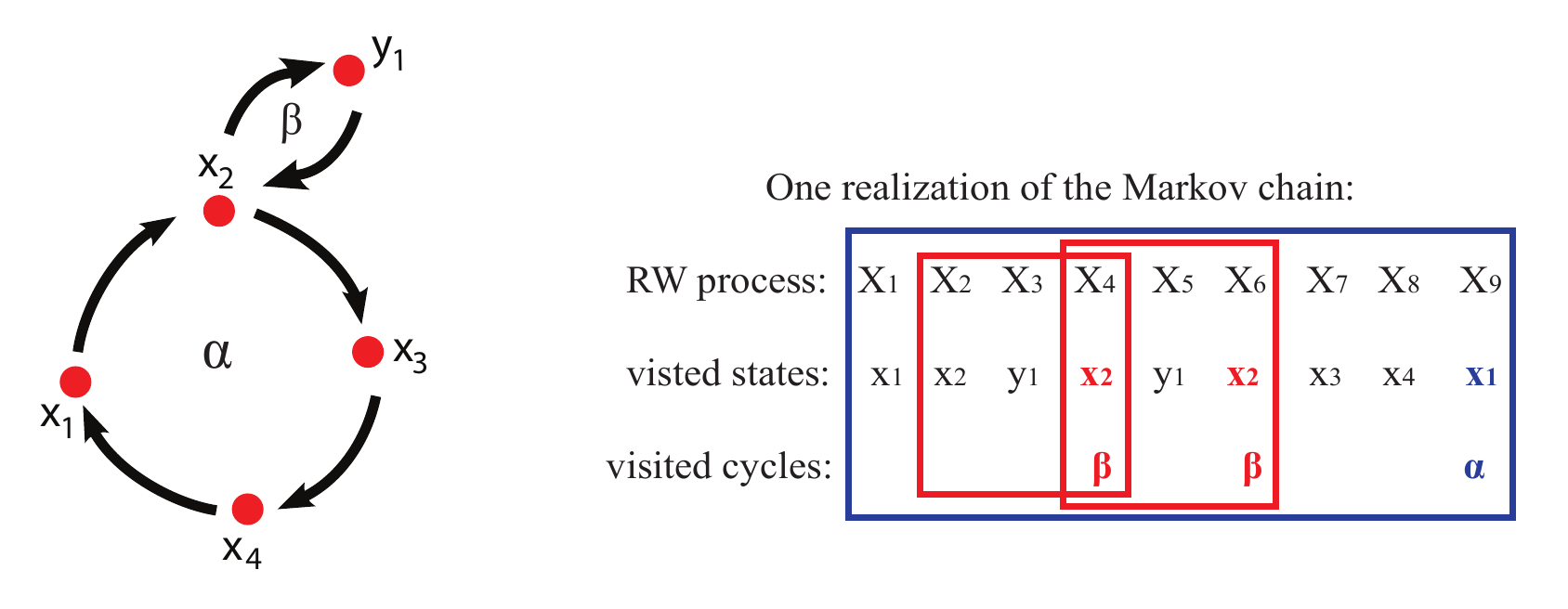}
   \caption{Examples of two simple cycles $\alpha$ and $\beta$ and one realization of the Markov chain $(X_n)_x$.}
  \label{fig:CycleEx}
    \end{center}
\end{figure}
\noindent
As a consequence of ergodicity and the flow conservation law (\ref{eq:consFlow}), we can find a cycle decomposition of $F$ into cycles. More precisely, there is a collection $\Gamma\subset\mathcal{C}$ of cycles $\gamma$ with real positive weights $F(\gamma)$ such that for every edge $(xy)$ the flow decomposition formula holds
\begin{equation}
F_{xy} = \sum_{\gamma\supset(xy)} F(\gamma),
\label{eq:flowdecomp}
\end{equation}
where we write $\gamma\supset (xy)$ if the edge $(xy)$ is in $\gamma$. Different approaches for finding cycle decompositions of Markov processes \cite{Schnakenberg76,Zia07,Kalpazidou2006,Timme12} have been developed. 
In this paper we will refer to a \textbf{stochastic decomposition approach} introduced in \cite{Kalpazidou2006,Qian2004}, as this approach leads to finding a unique cycle decomposition. Let $(X_n)_{1\leq n\leq T}$ be a realization of the Markov chain given by (\ref{eq:RWmatrix}).
We will say that $(X_n)_{1\leq n\leq T}$ \textbf{passes through the edge} $(xy)$ if there exists $k < T$ such that $X_k = x$ and $X_{k+1} = y$. The process $(X_n)_{1\leq n\leq T}$ \textbf{passes through a cycle} $\gamma$ if it passes through all edges of a cycle $\gamma$ in the right order, but not necessarily consecutively meaning that excursions through one or more full new cycles are allowed. This behavior is illustrated in Figure \ref{fig:CycleEx}.\\
Let $N_T^\gamma$ be the number of times $(X_n)_{1\leq n\leq T}$ passes through a cycle $\gamma$ up to time $T$. The limit
\begin{equation}
w(\gamma) := \lim_{T\rightarrow\infty} \frac{N_T^\gamma}{T}
\label{eq:WC}
\end{equation}
exists almost surely \cite{Qian2004}. Then, if the following \textbf{weights condition (WC)} holds
\begin{equation}
\mathbf{WC}:\; \Gamma = \{\gamma | w(\gamma)>0\}, \quad F(\gamma) = w(\gamma),
\end{equation}
we obtain a unique cycle decomposition $(\Gamma,w)$. In terms of the random walk process, the weight $w(\gamma)$ denotes the mean number of passages through a cycle $\gamma$ of almost all sample paths $(X_n)_{1\leq n\leq T}$. We will refer to $w(\gamma)$ as an importance weight for the cycle $\gamma$.\\
An explicit formula to calculate the weights $w(\gamma)$ was given in \cite{Qian2004}, which reads  $\gamma=(x_1,\ldots x_n)$
\begin{equation}
 w(\gamma) = \pi_{x_1}p_{x_1x_2}\ldots p_{x_nx_1} N_{x_1\ldots x_n}
 \label{eq:weightsformula}
\end{equation}
where the normalization factor $N_{x_1\ldots x_n}$ is a product of taboo Green's functions accounting for the excursions the random walk process can take while completing $\gamma$. It is hard to calculate $w(\gamma)$ in general, so we will use (\ref{eq:weightsformula}) for illustration purposes only, see \cite{Qian2004} for more details.\\
\begin{figure}[htp]
\begin{center}
 \includegraphics[width=\textwidth]{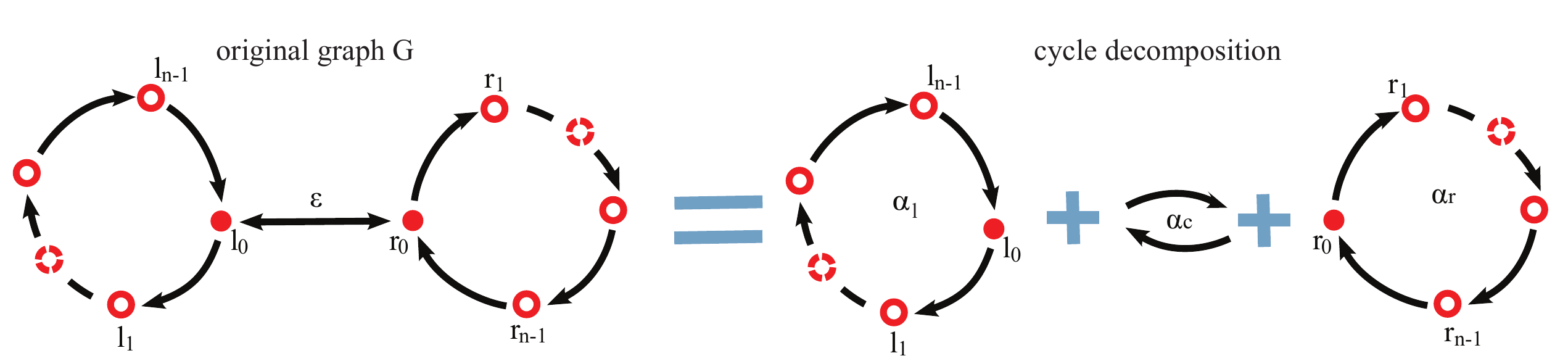}
   \caption{The barbell graph example network, consisting of two cycles with $n$ nodes joined by an edge with weight $\epsilon$. The cycle decomposition of the barbell graph consists of three cycles $\alpha_l$, $\alpha_c$ and $\alpha_r$.}
  \label{fig:barbell}
    \end{center}
\end{figure}
\noindent
As an example we will consider the barbell graph presented in Figure \ref{fig:barbell}. This is a weighted, directed graph consisting of two cycles $\alpha_l = (l_0, l_1, \ldots, l_{n-1})$ and $\alpha_r = (r_0, r_1,\ldots, r_{n-1})$ with $n$ nodes each and all edge weights are one. These two cycles are joined by an undirected edge $\alpha_c = (l_0, r_0)$ with edge weight $K(l_0,r_0) = K(r_0,l_0) = \eps<1$. Then, the probability of a standard random walk process (\ref{eq:RWmatrix}) leaving any of the two cycles at $l_0$ or $r_0$ is $p_{l_0 r_0} = \frac{\epsilon}{1+\epsilon} = p_{r_0 l_0}$. The central nodes have the same stationary distribution $\pi_{l_0} = \pi_{r_0} = \pi_c$ and for all other nodes we have $\pi_{l_i} = \pi_{r_i} = \pi_l$. Using the flow conservation property (\ref{eq:consFlow}) at one of the central nodes we get the equation $\pi_l + \pi_c\frac{\eps}{1+\eps} = \pi_c$ from which we can easily calculate
\[
\pi_l = \frac{1}{2(n+\eps)},\quad \pi_c = (1+\eps)\pi_l = \frac{1+\eps}{2(n+\eps)}.
\]
Since every edge belongs to exactly one of the three loops $\alpha_c = (l_0, r_0)$, $\alpha_l = (l_0, l_1, \ldots, l_{n-1})$ and $\alpha_r = (r_0, r_1,\ldots, r_{n-1})$, the weights of these loops can be inferred directly from (\ref{eq:flowdecomp}) as
\[
w(\alpha_l) = w(\alpha_r) = \frac{1}{2(n+\eps)}=: w, \quad w(\alpha_c) = \eps w.
\]

\subsection{Cycle-based random walks.}\label{sec:CycleRW}
In this section, we will use the stochastic cycle decomposition to introduce reversible processes that are a good approximation of the original non-reverisble process, in the sense that they contain all the important directional information needed for the module identification. This will allow us to construct graph transformations which will map a directed graph into an undirected graph, such that cycle node connections (i.e. if $x,y\in\alpha$) and the cycle importance weights are encoded into the structure of the new graph. Then, we will be able to reduce the original module finding problem on directed networks to a well understood problem of finding module partitions of undirected networks.\\
\noindent
We start by reformulating the original dynamics $(X_n)_n$ on vertices $V$ in terms of a dynamics on the cycle space $\Gamma$. Namely, instead of the dynamics of the standard random walk process with jumps from one node to another node, we will look at the random walk process with dynamics of the type node-cycle-node. To this end, we introduce
\begin{definition}
 The process $(\tilde X_n)_n$ is a stochastic process on $V$ generated by the following procedure:
\begin{enumerate}\label{def:cycle1}
 \item If $\tilde X_n=x$, then draw the next cycle $\xi_{n}$ according to the conditional probabilities
 \[
  \pP(\xi_n=\alpha|\tilde X_n=x) =: b_\alpha^{(x)} = \begin{cases}
                                                \frac{1}{\pi_x} w(\alpha), & x \in \alpha \\
                                                0 & x \notin \alpha.
                                              \end{cases}
 \]
 \item If $\xi_n = \alpha$, then follow the cycle $\alpha$ one step to obtain $\tilde X_{n+1} = y, (xy)\subset \alpha$.
\end{enumerate}
\end{definition}
\noindent
In this process, jumps from a node $x$ to a cycle $\alpha$ are probabilistic with probabilities $b_\alpha^{(x)}$ given by the cycle decomposition $\{\Gamma, w\}$. Using (\ref{eq:flowdecomp}) it follows that $\sum_{\alpha \supset x} w(\alpha) = \pi_x$, so that $b_\alpha^{(x)}$ form indeed conditional probabilities. Furthermore, $(\tilde X_n)_n$ is a Markov process with a transition matrix $P$
\begin{eqnarray*}
 \pP(\tilde X_{n+1} = y|\tilde X_n = x) & = & \sum_{\alpha\supset (xy)} \pP(\xi_n = \alpha | \tilde X_n = x) \\
 & = & \frac{1}{\pi_x} \sum_{\alpha\supset (xy)} w(\alpha) \\
 & \stackbin{\eqref{eq:flowdecomp}}{=} & \frac{1}{\pi_x} F_{xy} = p_{xy}.
\end{eqnarray*}
\noindent
Thus, $(X_n)_n$ and $(\tilde X_n)_n$ are statistically identical and we will refer to them as $(X_n)_n$ in the rest of the paper. As a byproduct we obtained the associated chain $(\xi_n)_n \subset \Gamma$ of cycles used by $(X_n)_n$. Notice that cycle-node jumps are deterministic and that $(\xi_n)_n$ is a non-Markovian process. To see this, consider the barbell graph in Figure \ref{fig:barbell}. We have $\pP(\xi_{n+2} = \alpha_r |\xi_{n+1} = \alpha_r, \xi_n = \alpha_c ) = 1$ since in this case $X_{n+2} = r_1$ with probability $1$. But $\pP(\xi_{n+2} = \alpha_r | \xi_{n+1} = \alpha_r) = 1 - \frac{1}{n}\frac{\eps}{1+\eps}$ as in this case $X_{n+2}=r_0$ with probability $1/n$, so that $X_{n+2} = \alpha_c$ is possible.\\
If $(X_n)_n$ is in equilibrium such that $\pP(X_n = x) = \pi_x$, then the chain $(\xi_n)_n$ is also in equilibrium, with distribution
\begin{eqnarray*}
 \mu(\alpha) &:= &\pP(\xi_n = \alpha) = \sum_x \pP(\xi_n = \alpha|X_n = x)\pP(X_n = x) = \sum_{x\subset \alpha} b_\alpha^{(x)} \pi_x \\
 & = &\sum_{x\subset \alpha} \frac{1}{\pi_x}w(\alpha)\pi_x = |\alpha|w(\alpha).
\end{eqnarray*}
Further assuming that $(X_n)_n$ is in equilibrium, the one-step transition matrix of $(\xi_n)_n$ is
\begin{equation}
 \mathcal{Q}_{\alpha\beta} := \pP\left(\xi_{n} = \beta| \xi_{n-1} = \alpha\right) = \sum_{x\in (\alpha\cap\beta)} \frac{1}{|\alpha|} b_\beta^{(x)},
 \label{eq:Ptilde}
\end{equation}
see Appendix for detailed calculation. Moreover, $\mathcal{Q}$ is detailed-balanced
\begin{equation}
 \mu(\alpha)\mathcal{Q}_{\alpha\beta} = \sum_{x\in (\alpha\cap\beta)}\frac{1}{\pi_x}w(\alpha)w(\beta) = \mu(\beta)\mathcal{Q}_{\beta\alpha}. \label{eq:loop_db}
\end{equation}
Thus, given a non-reversible Markov chain $(X_n)_n$ on $V$, we have constructed a reversible, but non-Markovian associated chain $(\xi_n)_n$ on the cycle space $\Gamma$. We will now modify Definition \ref{def:cycle1} to obtain two Markov chains $(Z_n)_n\subset V$ and $(\zeta_n)_n\subset\Gamma$. It will turn out later that $(\zeta_n)_n$ is the Markov approximation of $(\xi_n)_n$.
\begin{definition}
 The {\em cycle-node-cycle} process $(\zeta_n)_n \subset \Gamma$ and the {\em node-cycle-node} process $(Z_n)_n\subset V$ are stochastic processes generated by the following procedure:
 \begin{enumerate}
 \item If $Z_n=x$, then draw the next cycle $\zeta_n$ according to the conditional probabilities
 \[
 \pP(\zeta_n=\alpha|Z_n=x) = b_\alpha^{(x)}.
 \]
 \item If $\zeta_n = \alpha$, then draw $Z_{n+1} = y$ with probability
 \[
   \pP(Z_{n+1} = y|\zeta_n=\alpha) := \begin{cases}
                                      \frac{1}{|\alpha|}, & y\in \alpha \\
                                      0, & y\notin \alpha.
                                     \end{cases}
 \]
\end{enumerate}
\end{definition}
\noindent
Jumps node-cycle are again probabilistic, but now also jumps from a cycle to a node are probabilistic because from a cycle $\alpha$ the process can jump with uniform probability to any of the nodes in this cycle. Both processes $(Z_n)_n$ and $(\zeta_n)_n$ are Markov chains and we will now look at their transition properties. 
To simplify the notation we introduce the two stochastic matrices $\mathcal{B}: \R^\Gamma \rightarrow \R^V$ and $\mathcal{V}: \R^V \rightarrow \R^\Gamma$ with components
\begin{eqnarray}
 \mathcal{B}_{x\alpha} & = & \pP(\zeta_n=\alpha|Z_n=x) = b_{\alpha}^{(x)}, \notag\\
 \mathcal{V}_{\alpha y} & = & \pP(Z_{n+1} = y|\zeta_n=\alpha) = \frac{1}{|\alpha|}\delta(y\in \alpha). \label{eq:BandV}
\end{eqnarray}
Let $\pi_n(x) := \pP(Z_n=x)$ and $\mu_n(\alpha) := \pP(\zeta_n = \alpha)$ be the probability distributions of $(Z_n)_n$ and $(\zeta_n)_n$. Then the above definition implies the following propagation scheme for $\pi_n$ and $\mu_n$
\begin{equation}
 \pi_n \stackrel{\mathcal{B}^T}{\longrightarrow} \mu_n \stackrel{\mathcal{V}^T}{\longrightarrow} \pi_{n+1} \stackrel{\mathcal{B}^T}{\longrightarrow} \mu_{n+1} \stackrel{\mathcal{V}^T}{\longrightarrow} \ldots
 \label{eq:propscheme}
\end{equation}
The next lemma gives the explicit expressions for the transition matrices of $(\zeta_n)_n$ and $(Z_n)_n$ in terms of $\mathcal{B}$ and $\mathcal{V}$ 

 \begin{lemma}
  Let $\mathcal{P}$ be the transition matrix of $(Z_n)_n$ and $\mathcal{Q}$ be the transition matrix of $(\zeta_n)_n$. Then $\mathcal{P} = \mathcal{B}\mathcal{V}$ and $\mathcal{Q} = \mathcal{V}\mathcal{B}$ and
  \begin{eqnarray}
   \mathcal{P}_{xy} & = & \frac{1}{\pi_x}\sum_{\alpha\ni x,y} \frac{w(\alpha)}{|\alpha|} \label{eq:P} \\
   \mathcal{Q}_{\alpha\beta} & = & \sum_{x\in (\alpha\cap\beta)} \frac{1}{|\alpha|} b_\beta^{(x)}. \label{eq:Q}
  \end{eqnarray}
  \end{lemma}

\begin{proof}
By the Kolmogorov forward equation $\pi_{n+1} = \mathcal{P}^T\pi_n$, which in view of (\ref{eq:propscheme}) immediately implies $\mathcal{P} = \mathcal{B}\mathcal{V}$, and similarly we get $\mathcal{Q} = \mathcal{V}\mathcal{B}$. On the other hand, using (\ref{eq:BandV}):
\[
 \left(\mathcal{V}\mathcal{B}\right)_{\alpha\beta} = \sum_{x\in \alpha\cap\beta}b_\beta^{(x)} \frac{1}{|\alpha|}, \quad (\mathcal{B}\mathcal{V})_{xy} = \frac{1}{\pi_x}\sum_{\alpha\ni x,y} \frac{w(\alpha)}{|\alpha|}.
\]
Since $\mathcal{B}$ and $\mathcal{V}$ are stochastic, then $\mathcal{P}$ and $\mathcal{Q}$ are stochastic matrices too.
\end{proof}

\begin{table}[h]
\resizebox{\textwidth}{!}{\begin{minipage}{\textwidth}
 \begin{adjustwidth}{0.5cm}{}
    \begin{tabular}{ | c | c | c | c | c |}
    \hline
    \textbf{Process} & $\mathbf{(X_n)_n}$ & $\mathbf{(\xi_n)_n}$ & $\mathbf{(Z_n)_n}$ & $\mathbf{(\zeta_n)_n}$\\ \hline\hline
    Properties & non-reversible & reversible & reversible & reversible\\
            & Markov & non-Markov & Markov & Markov\\ \hline
    State space & $V$ & $\Gamma$ & $V$ & $\Gamma$\\ \hline
    Transition matrix & $P$ & $\mathcal{Q}$ & $\mathcal{P}$ & $\mathcal{Q}$\\ \hline
    Stationary distribution & $\pi$ & $\mu$ & $\pi$ & $\mu$\\ \hline
    \end{tabular}
\caption[Table caption text]{Properties of the introduced random walk processes.}
\label{table:tableP}
 \end{adjustwidth}
\end{minipage} }
\end{table}
\noindent
By comparing with (\ref{eq:Ptilde}), we can see that the one-step transition matrices of $(\zeta_n)_n$ and $(\xi_n)_n$ are the same, so that $(\zeta_n)_n$ is the Markov approximation of $(\xi_n)_n$. Thus, our construction produced a Markov process $(Z_n)_n$ on the space of graph nodes $V$ with transition matrix $\mathcal{P}$ and a closely related Markov process $(\zeta_n)_n$ on the space of graph cycles $\Gamma$ with the transition matrix $\mathcal{Q}$. Table \ref{table:tableP} shows the main properties of the introduced random walk processes.
To establish a more precise relation between these processes, we will examine the spectral properties of $\mathcal{P}$ and $\mathcal{Q}$. \\
\noindent
Let $D_\pi = \mbox{diag}(\pi)$ and $D_\mu = \mbox{diag}(\mu)$ with $\mu(\alpha) = |\alpha|w(\alpha)$ as before. We equip $\R^V$ and $\R^\Gamma$ with the scalar products $\langle u, v\rangle_\pi := u^T D_\pi v$ and $\langle u,v\rangle_\mu := u^T D_\mu v$. Then, Lemma \ref{lemma:PQprop} shows that $(Z_n)_n$ and $(\zeta_n)_n$ are two time-reversible Markov processes with identical spectral properties.

\begin{lemma}\label{lemma:PQprop}
The transition matrices $\mathcal{P}$ and $\mathcal{Q}$ of processes $(Z_n)_n$ and $(\zeta_n)_n$ have the following properties:
 \begin{enumerate}
  \item[(i)] $\pi$ is the stationary distribution of $\mathcal{P}$ and $\mu$ is the stationary distribution of $\mathcal{Q}$.
  \item[(ii)] $(Z_n)_n$ and $(\zeta_n)_n$ are time-reversible processes, i.e. $\langle u,\mathcal{P}v\rangle_\pi = \langle\mathcal{P} u,v\rangle_\pi$ and $\langle u,\mathcal{Q}v\rangle_\mu = \langle\mathcal{Q} u,v\rangle_\mu$.
  \item[(iii)] The non-zero spectrum of $\mathcal{Q}$ is identical to the non-zero spectrum of $\mathcal{P}$, i.e. $\sigma(\mathcal{Q})\setminus\{0\} = \sigma(\mathcal{P})\setminus\{0\}$.
  \item[(iv)] For $\lambda \in \sigma(\mathcal{P})\setminus\{0\}$, $\dim\ker (\mathcal{P} - \lambda 1) = \dim\ker (\mathcal{Q} - \lambda 1)$. If $v \in \ker (\mathcal{Q} - \lambda 1)$, then $\mathcal{B}v\in \ker(\mathcal{P} - \lambda 1)$. If $v \in \ker (\mathcal{P} - \lambda 1)$, then $\mathcal{V}v\in \ker(\mathcal{Q} - \lambda 1)$
 \end{enumerate}
\end{lemma}

\begin{proof}
See Appendix.
\end{proof}

\begin{figure}[htp]
 \begin{subfigure}[b]{0.40\textwidth}
 \includegraphics[width=\textwidth]{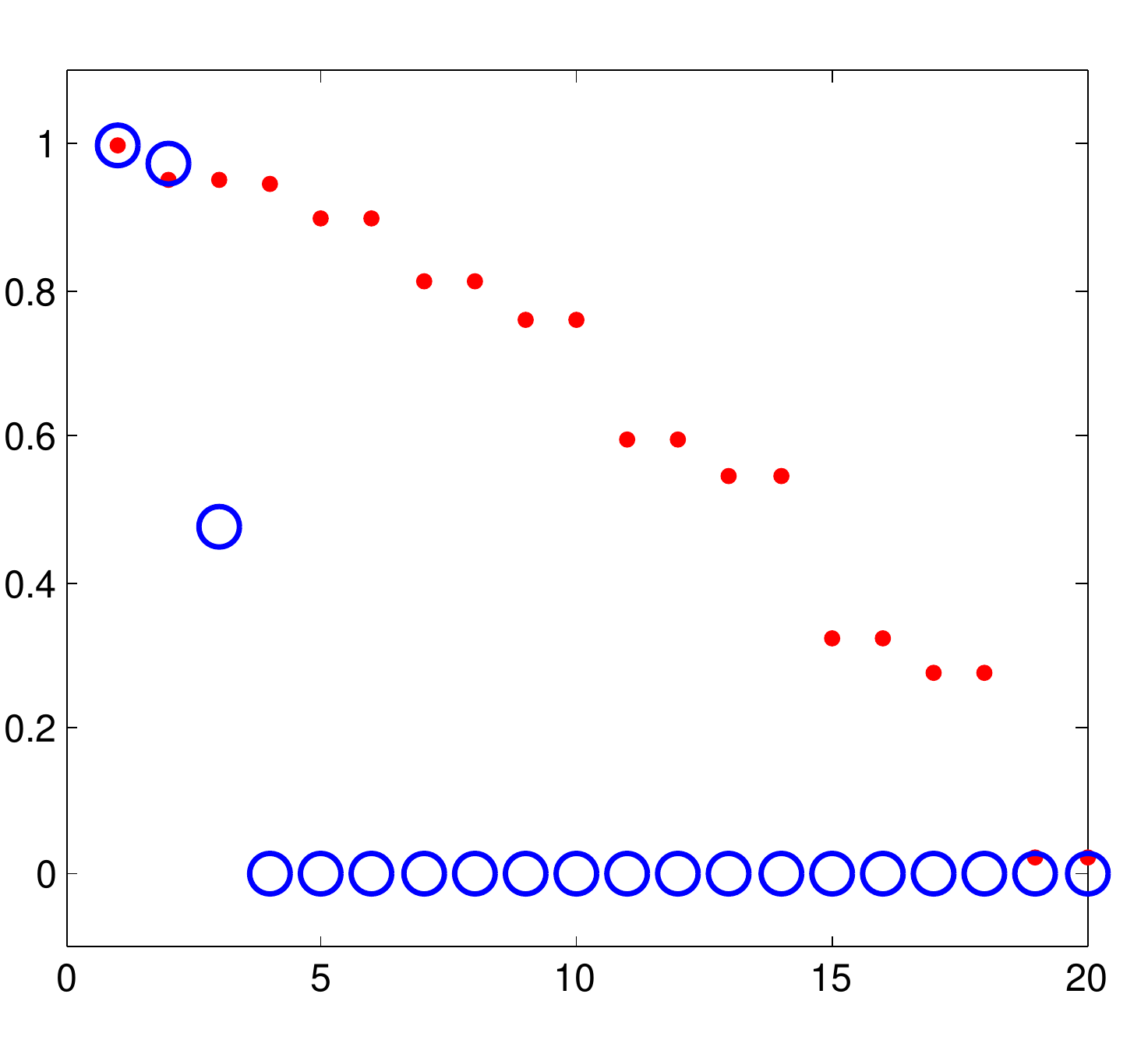}\caption{}\label{fig:BSpectrumReal}
 \end{subfigure}
 \begin{subfigure}[b]{0.40\textwidth}
 \includegraphics[width=\textwidth]{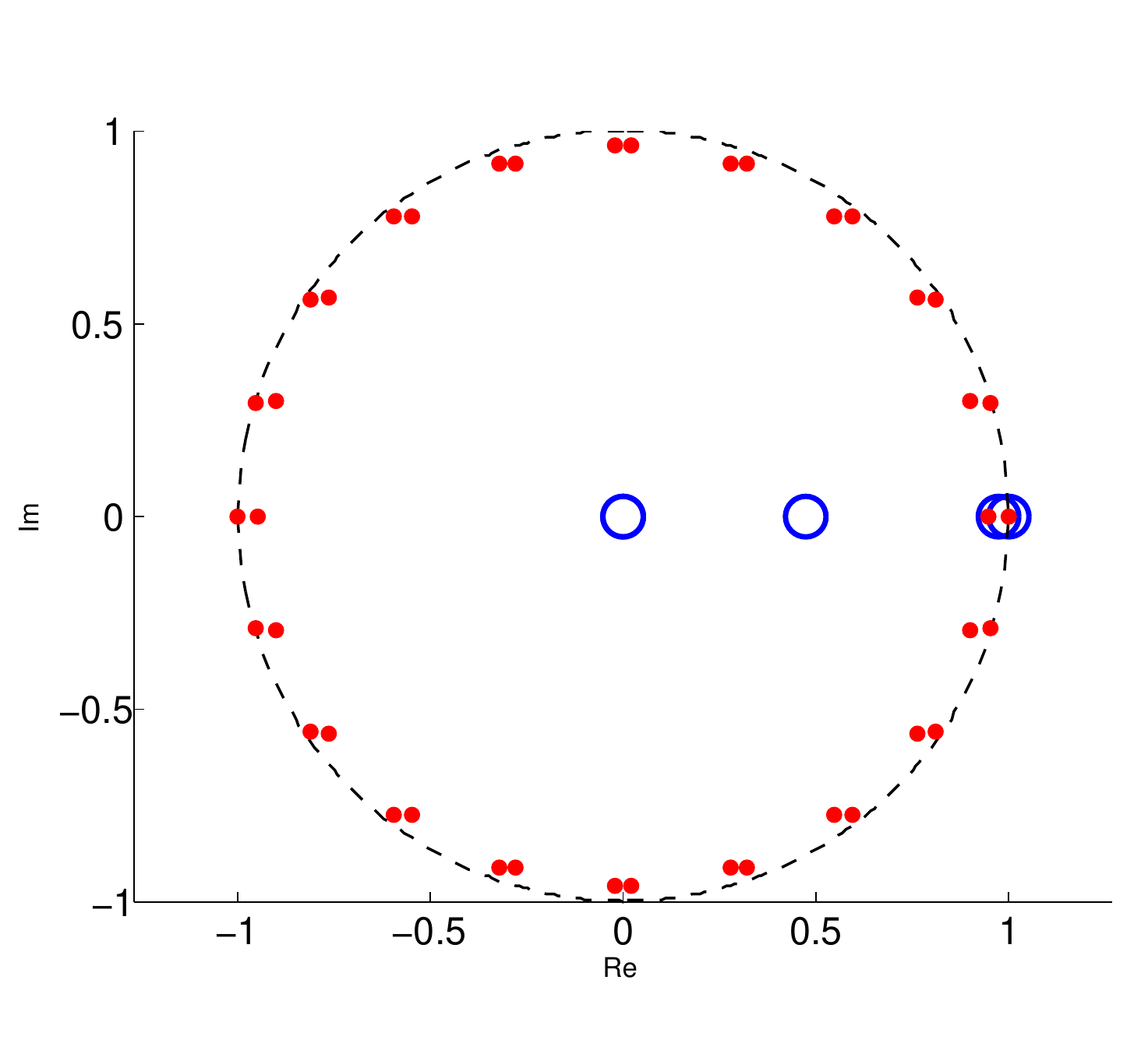}\caption{}
 \label{fig:BSpectrumAll}
 \end{subfigure}
   \caption{The eigenvalues of random walk processes on a barbell graph with $n = 40$ and $\epsilon=0.1$. Red: Eigenvalues of $P$. Blue: Eigenvalues of $\mathcal{P}$. \ref{fig:BSpectrumReal} shows the real part of the largest $20$ eigenvalues of $P$ and $\mathcal{P}$. \ref{fig:BSpectrumAll} shows the complete spectrum of $P$ and $\mathcal{P}$.}\label{fig:barbellSpectrum}
\end{figure}

Figure \ref{fig:barbellSpectrum} shows the spectrum of transition matrices $P$ and $\mathcal{P}$ for the barbell graph example from Figure \ref{fig:barbell} and $n = 40$,  $\epsilon=0.1$. The spectrum of the transition matrix $P$ of the original non-reversible random walk process consists of complex eigenvalues which all lie almost on the unit circle, i.e. have a module almost $1$. On the other hand, the spectrum of $\mathcal{P}$ (therefore also $\mathcal{Q}$) is real valued and offers a clear spectral gap after the second eigenvalue, indicating the existence of two network modules. The benefit of the new random walk process in this example is that it preserves the dominant real valued eigenvalues related to the existence of network modules and projects out the complex eigenvalues.

\subsection{Communication graph and cycle graph}\label{sec:cycledyn}
Time-reversibility of $(Z_n)_n$ and $(\zeta_n)_n$ directly implies that the weights $I_{xy} := \pi_x\mathcal{P}_{xy}$ and $W_{\alpha\beta} := \mu(\alpha)\mathcal{Q}_{\alpha\beta}$ are symmetric. Remarkably, both $W_{\alpha\beta}$ and $I_{xy}$ have an interpretation in terms of the original Markov process $(X_n)_n$. More precisely, the quantity
\begin{equation}
 W_{\alpha\beta} = \sum_{x\in (\alpha\cap\beta)} w(\alpha) b_\beta^{(x)},\quad \alpha,\beta\in\Gamma
 \label{eq:Wedge}
\end{equation}
is the total exchange of probability flux $F$ between the cycles $\alpha$ and $\beta$ per timestep. To see this, note that for $x\in \alpha$, $w(\alpha)$ is the fraction of probability current arriving at $x$ which is carried by $\alpha$, and $b_\beta^{(x)}$ is the probability to switch to $\beta$ at $x$, so that $w(\alpha)b_\beta^{(x)}$ is the exchange of probability flux between $\alpha$ and $\beta$ that happens at $x$. The quantity
\[
 I_{xy} = \sum_{\alpha\ni x,y} \frac{w(\alpha)}{|\alpha|}
\]
is a measure for intensity of communication between nodes $x,y$ under $(X_n)_n$ in terms of the cycles through which $x$ and $y$ are connected. Indeed, $I_{xy}$ is large, when $x$ and $y$ are connected by many cycles $\alpha$ which are important (i.e. $w(\alpha)$ large) and short (i.e. $|\alpha|$ small). This interpretation of $I_{xy}$ as a communication intensity measure will be made clear in the next Section \ref{sec:node_comm}.\\
\noindent
We will use these properties of $W_{\alpha\beta}$ and $I_{xy}$, to introduce the following undirected graphs:

\begin{definition}
 Given the directed graph $G = (V,E)$, we let $H_G(\Gamma, E_W)$ be the undirected graph with vertex set $\Gamma$ which has an edge with weight $W_{\alpha\beta}$ between $\alpha$ and $\beta$ if $W_{\alpha\beta} > 0$, and we call $H_G$ the \textbf{cycle graph} of $G$. We let $K_G(V,E_I)$ be the undirected graph with vertex set $V$ which has an edge with weight $I_{xy}$ between $x$ and $y$ if $I_{xy} > 0$, and we call $K_G$ the \textbf{communication graph} of $G$.
\end{definition}
Notice that $H_G$ is simply the network representation of the cycle-node-cycle process $(\zeta_n)_n$ and $K_G$ is the network representation of the node-cycle-node process $(Z_n)_n$. Thus, the standard random walk processes on $H_G$ and $K_G$ will have transition matrices $\mathcal{Q}$ and $\mathcal{P}$ respectively. Since important module information from $G$ is encoded in the structure of both $H_G$ and $K_G$, we can apply some of the existing module finding approaches in undirected graphs on $H_G$ and $K_G$ and project the result to $G$. In this paper we will present the algorithm for module identification which uses only the communication graph (see Section \ref{sec:Three}), because $H_G$ is usually much larger than $K_G$ for modular directed networks $G$.\\
\begin{figure}[htp]
\begin{center}
  \includegraphics[width=4in]{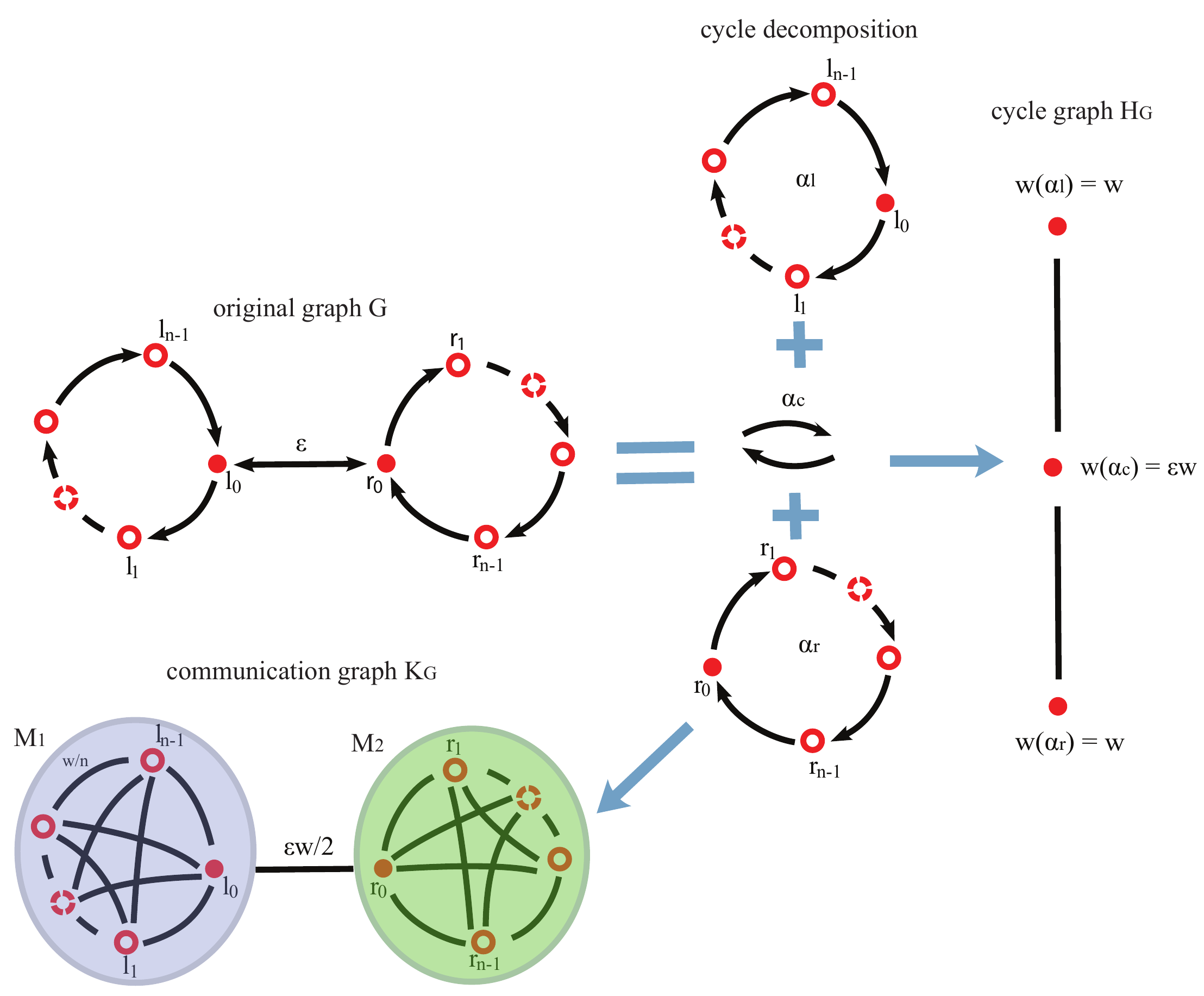}\\
  \caption{The barbell graph and corresponding cycle and communication graphs.}\label{fig:barbellAll}
  \end{center}
\end{figure}
In Figure \ref{fig:barbellAll}, the barbell graph is shown together with the corresponding cycle graph $H_G$ and communication graph $K_G$. Since the cycle decomposition of the barbell graph produces three cycles $\alpha_l$, $\alpha_c$ and $\alpha_r$, the cycle graph will have only three nodes. The edge weights (\ref{eq:Wedge}) of the cycle graph are $W_{\alpha_l\alpha_c} = \frac{w\epsilon}{1+\epsilon} = W_{\alpha_c\alpha_r}$. The communication graph of the barbell graph has the following edge weights
\begin{equation}
I_{xy}  = \begin{cases} w/n & x,y\in \alpha_l\\
                                                              w/n & x,y\in \alpha_r\\
                                                              \eps w/2 & x\neq y\in \alpha_c \\
                                                              \eps w/2 + w/n & x=y \in \alpha_c \\
                                                              0 & \mbox{otherwise}
                                                             \end{cases}
                                                             \label{eq:Ibarbell}
\end{equation}
As observed in Figure \ref{fig:barbellAll}, this results in cycles $\alpha_l$ and $\alpha_r$ being mapped into two complete subgraphs with equally weighted edges and consequently the appearance of two modules $M_1$ and $M_2$ in $K_G$.\\

\begin{remark} (Line-graphs)
Cycle graphs can be related to line-graphs (or edge-graphs), which are important objects in spectral graph theory \cite{LineGraphs03,Cvetkovic04}. If $G =(V,E)$ is an undirected graph with $|E|$ edges, its line-graph $L(G)$ has $|E|$ nodes which are edges of $G$ and nodes in $L(G)$ are adjacent if edges in $G$ are. So, a transformation of $G$ into $L(G)$ maps graph's edges into nodes and a transformation of $G$ into $H_G$ maps graph's cycles into nodes. 
Furthermore, line-graphs can be very efficiently used for fuzzy clustering of undirected networks \cite{Evans09} by employing some of the full partitioning approaches for finding modules in line-graphs and then projecting these modules back to the original graph. Similarly, one can use cycle graphs for fuzzy partitioning of directed networks.
\end{remark}

\begin{remark} (Entropy production)
The work of \cite{Timme12} introduces cycle graphs as a way to describe non-equilibrium Markov processes. Indeed, many non-equilibrium properties may very well be related to cycle decompositions of the type we discussed here. We give one example. The {\em entropy production rate} is believed to be a measure of the degree of irreversibility of a system. For Markov chains, it is given by
\[
e_p = \sum_{i,j}\frac{1}{2}\left(\pi_i p_{ij} - \pi_j p_{ji}\right)\log\left(\frac{\pi_i p_{ij}}{\pi_j p_{ji}}\right).
\]
It has been shown in \cite{Qian2004} that the entropy production rate can be expressed by means of the stochastic cycle decomposition (\ref{eq:WC}) as
\[
e_p = \sum_{\gamma\in\Gamma}(w(\gamma) - w(\gamma^-))\log\left(\frac{w(\gamma)}{w(\gamma^-)}\right).
\]
Here, $\gamma^-$ denotes the cycle $\gamma$ with reversed orientation. This shows that the entropy production rate is given in terms of the difference of the cycle weights for cycles $\gamma$ and their reversed counterparts $\gamma^-$.  More precise relations between cycle decomposition and entropy production will be the topic of our research in the future.
\end{remark}

\subsection{Information transport}\label{sec:node_comm}

In this section we show that the cycle-node-cycle process $(Z_n)_n$ in fact mimics a model for information transport in the network, which gives an interpretation to the quantity $I_{xy}$. Suppose that initially, $X_0 = x$. Let $\tilde Z_0 = x$ denote the initial location of one bit of information. We will now model how this bit is transported through the network by means of $(X_n)_n$:

\begin{enumerate}
 \item Let $\tau$ be the first time $(X_n)_n$ returns to $x$. Since we assumed $(X_n)_n$ to be irreducible, $\tau$ is a.s. finite and the path $p = (X_0,\ldots, X_\tau)$ is a loop, possibly with self-intersections. We can decompose $p$ into simple cycles, exactly one of which contains $x$. Let this cycle be $\gamma$.
 \item The bit of information is now distributed uniformly among all nodes on $\gamma$. To model this, we will pick the new location $\tilde Z_1 = y$ of the bit uniformly at random among the vertices of $\gamma$. We will then restart the process $(X_n)_n$ at $y$ and repeat.
\end{enumerate}

Now observe that the outcome of step 1 is the cycle $\gamma$ exactly iff $(X_n)_n$ \textbf{passes through} $\gamma$, as it was defined in Section \ref{sec:flow}. The probability of the outcome being $\gamma$ is therefore the probability of $(X_n)_n$ passing through $\gamma$ conditioned on starting at $x$, which is $\delta(x\in \gamma)w(\gamma)/\sum_{\gamma\supset x} w(\gamma) = \mathcal{B}_{x\gamma}$. On the other hand, if the outcome of (1) is $\gamma$, then the probability of $\tilde Y_1 = y$ in (2) is $\frac{1}{|\gamma|}\delta(y\in \gamma) = \mathcal{V}_{\gamma y}$. We thus have by Lemma \ref{lemma:PQprop}:
\[
 \pP(\tilde Z_1 = y| \tilde Z_0 = x) = \sum_\gamma \mathcal{B}_{x\gamma}\mathcal{V}_{\gamma y} = \mathcal{P}_{xy} = \pP(Z_1 = y| Z_0 = x).
\]
This implies that $(\tilde Z_n)_n = (Z_n)_n$, that is the node-cycle-node process describes the sequence of locations of the bit of information which is transported along the network by $(X_n)_n$. Then
\[
 I_{xy} = \pP(Z_n = x, Z_{n+1} = y) = \sum_{\alpha\ni x,y} \frac{w(\alpha)}{|\alpha|}
\]
is the probability that the bit is delivered from $x$ to $y$, and therefore measures how easy it is to share information between the nodes $x$ and $y$. From the point of view of the network $G$, $I_{xy}$ connects topological properties of $G$ (how many cycles connect $x$ and $y$, and how long are they) with dynamical properties of the process $(X_n)_n$ (which gives the importance weights $w(\alpha)$). As we established in Section \ref{sec:cycledyn}, $I_{xy}$ is symmetric since $(Z_n)_n$ is reversible. This allows us to apply any clustering method designed for undirected, weighted networks to find a partition of $G$ into modules based on node communication. The algorithmic aspect of this approach will be presented in the next section together with a view on its computational complexity.

\section{Algorithmic aspects}\label{sec:Three}

\subsection{Identification of modules}

Following the ideas developed in Section \ref{sec:Three}, we will say that a subset $M\subset V$ is a module of the directed graph $G=(V,E)$ iff it is a module in the usual sense of the undirected, weighted {\em communication graph} $K_G = (V,E_I)$ defined in Section \ref{sec:cycledyn}. This suggests the following procedure to identify modules in $G$:

\begin{enumerate}
 \item Given the graph $G=(V,E)$ with edge weights $K(x,y)$, construct the transition matrix $P$ according to (\ref{eq:RWmatrix}).
 \item Obtain the cycle decomposition $\{\Gamma,w\}$ of the RW given by $P$, where the weights $w(\gamma)$ satisfy (\ref{eq:WC}).
 \item Construct the undirected weighted graph $K_G$ and find modules in $K_G$.
\end{enumerate}
\noindent
To partition $K_G$ in step (3), we will use Markov State Modeling (MSM) clustering \cite{Djurdjevac2011,SarichNet11}, but any other method designed for clustering undirected, weighted graphs could be used instead \cite{Fortunato2010,Newman2003b}. MSM clustering returns $m$ module cores $M_1,\ldots, M_m\subset V$, which are disjoint and do not necessarily form a full partition
\[
\cup_{i=1}^m M_i\not= V\quad \Rightarrow\quad T = V\setminus \cup_{i=1}^m M_i\not= \emptyset.
\]
The number of modules $m$ can be inferred from the spectrum of $P$, see \cite{Djurdjevac2010} for more details. Furthermore, the MSM approach aims at finding \textbf{fuzzy network partitions} into modules, such that some nodes do not belong to only one of the modules but to several or to none at all. For nodes that do not deterministically belong to one of the modules, we will say that they belong to the \textbf{transition region} $T$. The next obvious question is how to characterize the affiliation of the nodes from the transition region to one of the modules. In the MSM approach, this affiliation is given in terms of the random walk process $(Z_n)_n$ on $K_G$ by the \textbf{committor functions}
\[
q_i(x) = \pP\left[\tau_x\left(\cup_{j=1\atop j\not=i}^m M_j\right)>\tau_x(M_i)\right],\qquad i = 1,\ldots,m, \, x\in V
\]
where $\tau_x(A)$ is the first hitting time of the set $A\subset V$ by the process $(Z_n)_n$ if started in  node $x$, $\tau_x(A)=\inf\{t\ge 0: Z_n\in A, \,Z_0=x\}$. Therefore, $q_i(x)$ gives us the probability that the random walk process if started in node $x$, will enter the module $M_i$ earlier than any other module. Committor functions can be computed very efficiently by solving sparse, symmetric and positive definite linear systems \cite{MeSchEve,Djurdjevac2010}
\begin{eqnarray}
(\mathrm{Id}-P)q_i(x) & = & 0,\quad \forall x\in T\nonumber\\
q_i(x) & = & 1,\quad \forall x\in M_i\label{committor}\\
q_i(x) & = & 0,\quad \forall x\in M_j, j\neq i.\nonumber
\end{eqnarray}
Furthermore, it is easy see that the committor functions $q_1,\ldots,q_m$ form a partition of unity
\[
\sum_{i=1}^m q_i(x)=1,\quad\forall x\in V,
\]
such that we can interpret $q_i(x)$ as the natural random walk based probability of affiliation of a node $x$ to a module $M_i$.

\subsection{Obtaining the cycle decomposition}

We now turn our attention to the problem of actually obtaining the cycle decomposition $\{\Gamma,w\}$ with weights $w(\gamma)$ satisfying \eqref{eq:WC}. This is clearly a hard problem: While even finding and enumerating all cycles in $G$ scales exponentially in the worst case, explicitly computing the weight $w(\gamma)$ for $\gamma = (x_1,\ldots x_l)$ according to (\ref{eq:weightsformula}), requires the computation of $l$ taboo Green functions of the form $(I - P_{\{x_1,\ldots,x_s\}})^{-1}$ for $1\leq s\leq l$
where $P_{\{x_1,\ldots,x_s\}}$ is $P$ with rows and columns indexed by $\{x_1,\ldots,x_s\}$ deleted \cite{Qian2004}. This will not be feasible for large networks. We therefore resort to a sampling algorithm which computes the collection $\Gamma$ of cycles and counts $N_T^\gamma$ from a realization $(X_n)_{1\leq n\leq T}$ of the Markov chain $(X_n)_n$, see \cite{Qian2004}. The algorithm uses an auxiliary chain $\eta$ which is initialized at $\eta = X_1$, then the states $X_i$ are added sequentially to $\eta$ until $\eta$ contains a cycle, which is then removed from $\eta$:

\begin{center}
\colorbox{boxback}{
  \begin{minipage}[t]{12cm}
\begin{center}
  \begin{minipage}[t]{11cm}
\vspace{0.5cm}
\textbf{Algorithm summary of the stochastic cycle decomposition:}\\
\textbf{(i)} Initialization: Set $\Gamma = \emptyset$ and $\eta = X_1$.\\
\textbf{(ii)} for $i=2$ to $T$: If $X_i \notin \eta$, set $\eta = [\eta, X_i]$. If $\eta = [\eta_1, \ldots, \eta_{p-1},X_i,\eta_{p+1},\ldots,\eta_l]$ then
\begin{enumerate}
\item Set $\gamma = [X_i,\eta_{p+1},\ldots, \eta_{l}]$ and $\eta = [\eta_1,\ldots,\eta_{p-1}]$.
\item If $\gamma\in \Gamma$, then $N_T^\gamma = N_T^\gamma + 1$. Otherwise add $\gamma$ to $\Gamma$ and set $N_T^\gamma=1$.
\end{enumerate}
  \end{minipage}
\vspace{0.5cm}
\end{center}
 \end{minipage}
}
\end{center}

We then set $w_T(\gamma) := \frac{1}{T}N^\gamma_T$. In \cite{Qian2004} it has been shown that $w_{T}(\gamma) \rightarrow w(\gamma)$ a.s. as $T\rightarrow \infty$, therefore the finite-time weights $w_{T}$ returned by the sampling algorithm are an approximation of the true weights $w$. The computational complexity of this algorithm is at most quadratic in $T$. However, the number of cycles with significant weights $w(\gamma)$ which we have to sample accurately might scale unfavorably with the size of the network. In practise, one has to choose $T$ such that convergence is observed. If the network is very large and very modular such that $(X_n)_n$ is very metastable, this might be prohibitively expensive.

\begin{remark}[Timeseries perspective]
In many applications, the network $G$ will be a representation of a finite timeseries of data $(X_n)_{1\leq n\leq T}$. In such situations it is natural to use the sampling algorithm directly on the data, which fixes $T$ and removes any convergence issues. Our method then turns into a form of timeseries analysis via recurrence networks. We will address this point of view in more detail in further publications.
\end{remark}

\begin{remark}[Alternative Algorithms]
Various alternative algorithms exist to obtain a cycle decomposition $\{\tilde\Gamma, \tilde w\}$ which satisfies only (\ref{eq:flowdecomp}), but not (\ref{eq:WC}). We quickly review two examples:

\begin{enumerate}
\item[(a)] \textbf{Deterministic iterative algorithm:}\cite{Kalpazidou2006,Qian2004} This algorithm iterates the following steps: \\
(i) find a cycle $\gamma$ in $G$, \\
(ii) set $w(\gamma) = \min_{(xy)\subset \gamma} F_{xy}$,\\
(iii) update $F\rightarrow F - w(\gamma)\chi_\gamma$ where $\chi_\gamma$ is the unit flow along $\gamma$. \\
The iteration stops when $F<TOL$ is reached for a prescribed small $TOL>0$. This algorithm has polynomial complexity in the number of vertices of $G$.
\item[(b)] \textbf{Cycle-basis construction:} \cite{Kalpazidou2006} The space $\mathcal{C}$ of cycle flows is a vector space of dimension\footnote{The upper bound on $d$ is the Betti number of $G$.} $d \leq |E| - |V| + 1$ and can be equipped with a scalar product. Then we may take $\tilde\Gamma$ to be a basis of $\mathcal{C}$, which can be found e.g. via a minimal spanning tree \cite{Kalpazidou2006}. The weights $\tilde w$ can be computed by solving a linear system involving the $d\times d$ matrix $C_{\gamma\gamma'} = \langle \chi_\gamma, \chi_{\gamma'}\rangle$.
\end{enumerate}
If we use (a) or (b) instead of the sampling algorithm in step (2), then the resulting weights $\tilde I_{xy}$ loose their interpretation in terms of information transport measure between $x$ and $y$. Still, in many cases $\tilde I_{xy}$ might be a good approximation to $I_{xy}$, and (a) and (b) are fast even for large networks. The approximation quality of $\tilde I_{xy}$ obtained by (a) and (b) will be addressed in our future work.
\end{remark}

\subsection{Verification of modularity}

We can use the well-known modularity function \cite{Newman2006a} to evaluate the resulting partition $\{\chi_m\}_{m=1}^M$ of $K_G$. We obtain that
\begin{eqnarray}
 \bar Q & = & \sum_m\sum_{x,y} \chi_m(x)\left[ \pP(Z_n = x, Z_{n+1}=y) - \pP(Z_n=x)\pP(Z_n=y)\right]\chi_m(y)\notag\\
        & = & \sum_m\sum_{x,y} \chi_m(x)\left[I_{xy} - \pi_x\pi_y\right]\chi_m(y).\label{eq:Qbar}
\end{eqnarray}
Notice that $\bar Q$ is the modularity of the partitioning evaluated on the {\em undirected} graph $K_G$. Recently the extensions of modularity to directed networks have been discussed in literature \cite{Newman2008,Kim2010}. These approaches lead to the following modularity function $Q$ of partitioning on $G$
\begin{equation}
 Q = \sum_m\sum_{x,y} \chi_m(x)\left[\pi_x p_{xy} - \pi_x\pi_y\right]\chi_m(y).\label{eq:Qmod}
\end{equation}
However, approaches based on such $Q$ are shown to often neglect important directional information \cite{Fortunato2009,Schaub2012} and have the same outcome as when symmetrizing $P$. More precisely, $Q$ is unchanged if $\pi_xp_{xy}$ is replaced by the symmetrized version $\pi_xp^s_{xy} = \frac{1}{2}(\pi_xp_{xy} + \pi_yp_{yx})$, which ignores directions. In contrast, using $\bar Q$ to evaluate the quality of a partitioning of $K_G$ is unambiguous, and the directional information of $G$ is directly built into $I_{xy}$. We will demonstrate on some examples in Section \ref{sec:Four} that $Q$ and $\bar Q$ may differ significantly. This difference comes mainly from the fact that $Q$ takes into account only one-step transitions $p_{xy}$, unlike $\bar{Q}$ which considers multi-step transitions encoded in $I_{xy}$ via cycle structures of $G$.

\section{Numerical experiments}\label{sec:Four}

In this section, we will compare the results of three different clustering algorithms on a few example networks. The algorithms considered are:

\begin{enumerate}
 \item The algorithm proposed in Section \ref{sec:Three}, i.e. construction of $K_G$ via a realization $(X_n)_{1\leq n\leq T}$ and then MSM clustering of $K_G$, which we will refer to as {\bf Cycle MSM (CMSM) clustering}. This algorithm returns a fuzzy partition $\{\chi_m\}_{m=1}^M$.
 \item Construct $K_G$ via a realization $(X_n)_{1\leq n\leq T}$ and then optimize the $K_G$-based modularity function $\bar Q$ \eqref{eq:Qbar} over all full partitions $\{\chi_m\}_{m=1}^M$, $\chi_m(x)\in \{0,1\}$. We will refer to this as {\bf $\bar Q$-maximization}.
 \item Maximize the $G$-based modularity function $Q$ \eqref{eq:Qmod} over all full partitions $\{\chi_m\}_{m=1}^M$, $\chi_m(x)\in \{0,1\}$. This direct approach doesn't need the construction of $K_G$. We will refer to it as {\bf $Q$-maximization}.
\end{enumerate}

\subsection{Barbell graph}

As discussed earlier, the communication graph $K_G$ consists of two complete graphs with $n$ nodes each which are joined by the edge $(l_0r_0)$. CMSM clustering produces a full partition $\{\chi_1,\chi_2\}$ where $\chi_1$ is one on $C_1 = \alpha_l$ and zero on $C_2 = \alpha_r$, see Figure \ref{fig:barbellDecomposition}. The scores assigned by $Q$ and $\bar Q$ to $\{\chi_1,\chi_2\}$ are almost identical:
\[
Q(C_1,C_2) = \frac{1}{2} - \frac{\eps}{n+\eps}, \quad \bar Q(C_1,C_2) = \frac{1}{2} - \frac{1}{2}\frac{\eps}{n+\eps}.
\]

Now we address the question whether $\{\chi_1,\chi_2\}$ optimizes $Q$ resp. $\bar Q$. To do so, we compare $\{\chi_1,\chi_2\}$ for even $n$ with another partition into 3 sets $C_{1,a}, C_{1,b}$ and $C_2$ where $|C_{1,a}| = |C_{1,b}|$, see Figure \ref{fig:barbellDecomposition}. Let $\Delta Q = Q(C_1,C_2) - Q(C_{1,a},C_{1,b},C_2)$ and $\Delta \bar Q$ similarly. One finds
$$\Delta Q = \frac{1}{8} - \frac{1}{n+\eps}, \quad \Delta \bar Q = \frac{1}{8} - \frac{1}{4}\frac{n}{n+\eps}.$$

Then, $\Delta Q >0$ for $n\geq 8$ and $\Delta \bar Q < 0 $ as long as $n > \eps$. That is, as $n$ grows $Q$-maximization produces a partition into more and more subchains with less then 8 nodes while $\bar Q$-maximization always produces the partition $\{\chi_1,\chi_2\}$. Notice that the block structure (\ref{eq:Ibarbell}) indicates that it is equally easy to transmit information between all nodes in $C_1$ while it is hard to transmit information between $C_1$ and $C_2$. From the point of view of the RW process this is natural: The average time to go from $C_1$ to $C_2$ is $(n+\frac{1}{2})\frac{\eps}{1+\eps}$, while mixing within $C_1$ and in particular transitions between $C_{1,a}$ and $C_{1,b}$ happen quickly.

To summarize, the barbell graph has a modular structure which is not based on link density and therefore cannot be detected by density based methods such as $Q$-maximization. Rather, this modular structure is information based and revealed by $K_G$. In the particular case here, using CMSM clustering or $\bar Q$-optimization to partition $K_G$ yields the same result.

\begin{figure}[!ht]
    \centering
    \includegraphics[width=0.5\textwidth]{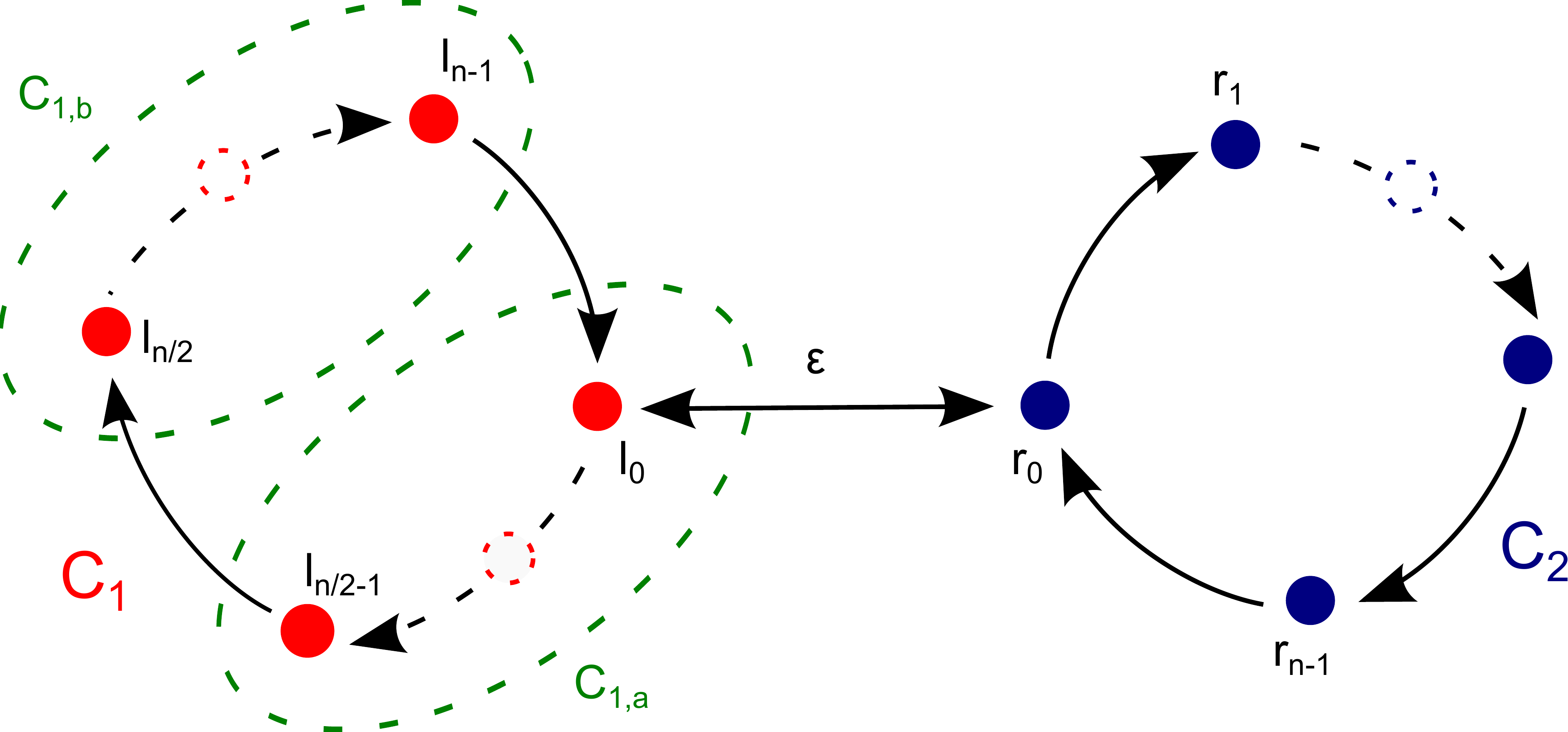}
    \caption{The two modules $C_1$ and $C_2$ produced by CMSM clustering. $Q$-maximization gives a partitioning into subchains of the type $C_{1,a}, C_{1,b}$.}
    \label{fig:barbellDecomposition}
\end{figure}

\subsection{Wheel switch}

Our next example is the 'wheel switch' graph which is shown in Figure \ref{fig:wheel_a}. It consists of a clockwise outer loop $\alpha_o$ and an anticlockwise inner loop $\alpha_i$ with $n$ nodes each and some connections between the two loops. At each connection, the probability to switch between $\alpha_o$ and $\alpha_i$ is $p$ and the connections are such that two smaller loops $\beta_1$ and $\beta_2$ of length 4 are formed (coloured in Figure \ref{fig:wheel_a}).

The community structure of this graph depends on $p$. If $p>1/2$, completing the small loops $\beta_1$ and $\beta_2$ is more likely then completing $\alpha_i$ and $\alpha_o$, so we expect a community structure where $\beta_1$ and $\beta_2$ dominate. Indeed, CMSM clustering produces the partition shown in Figure \ref{fig:wheel_a} which identifies $\beta_1$ and $\beta_2$ as modules and gives all other nodes an affiliation of $0.5$ to both modules (indicated by the white colour). The modularity values achieved by this partitioning depend on $p$ and $n$, but typical values are $\bar Q \approx 0.2$.

If $p\leq 1/2$, transitions between $\alpha_i$ and $\alpha_o$ are unlikely and we expect to find the community structure shown in Figure \ref{fig:wheel_b}. Indeed CMSM clustering produces exactly this partition. As in the previous example, for $p\leq 1/2$ $\bar Q$-maximization gives the same result. On the other hand and again as in the previous example, $Q$-maximization divides the red and green modules in Figure \ref{fig:wheel_b} further into many small chains of length less then $8$.

Also for $1/2 < p\leq 0.63$, $\bar Q$-maximization produces the partition shown in Figure \ref{fig:wheel_b}, while for $p>0.63$ it gives the one shown in Figure \ref{fig:wheel_c}. Note that the difference between Figure \ref{fig:wheel_a} and Figure \ref{fig:wheel_c} is only that the white nodes are now fully assigned to one of the two modules in such a way that as much of $\alpha_i$ and $\alpha_o$ is intact as possible. $Q$-maximization again produces many small chains of length less then 8 as modules, how many depends on $n$. Even if we restrict $Q$-maximization to partitions with only two modules, the optimum is degenerate: Any cut $S$ that produces modules of equal size and leaves $\beta_1$ and $\beta_2$ intact, as shown in Figure \ref{fig:wheel_d}, will optimize $Q$.

In summary, CMSM clustering exactly reproduces the expected community structure for all values of $p$ and $n$. $Q$-maximization destroys the community structure for large $n$ and produces many small chains instead. $\bar Q$-maximization cannot produce the fuzzy partition of Figure \ref{fig:wheel_a}, but instead produces the best possible full partition, independent of $n$. However, the $p$-threshold between outcomes \ref{fig:wheel_a} and \ref{fig:wheel_b} is $p=0.5$ for CMSM clustering, while it is $p=0.63$ between outcomes \ref{fig:wheel_b} and \ref{fig:wheel_c} for $\bar Q$-maximization.

\begin{figure}[h!]
 \begin{subfigure}[b]{0.45\textwidth}
 \includegraphics[width=0.95\textwidth]{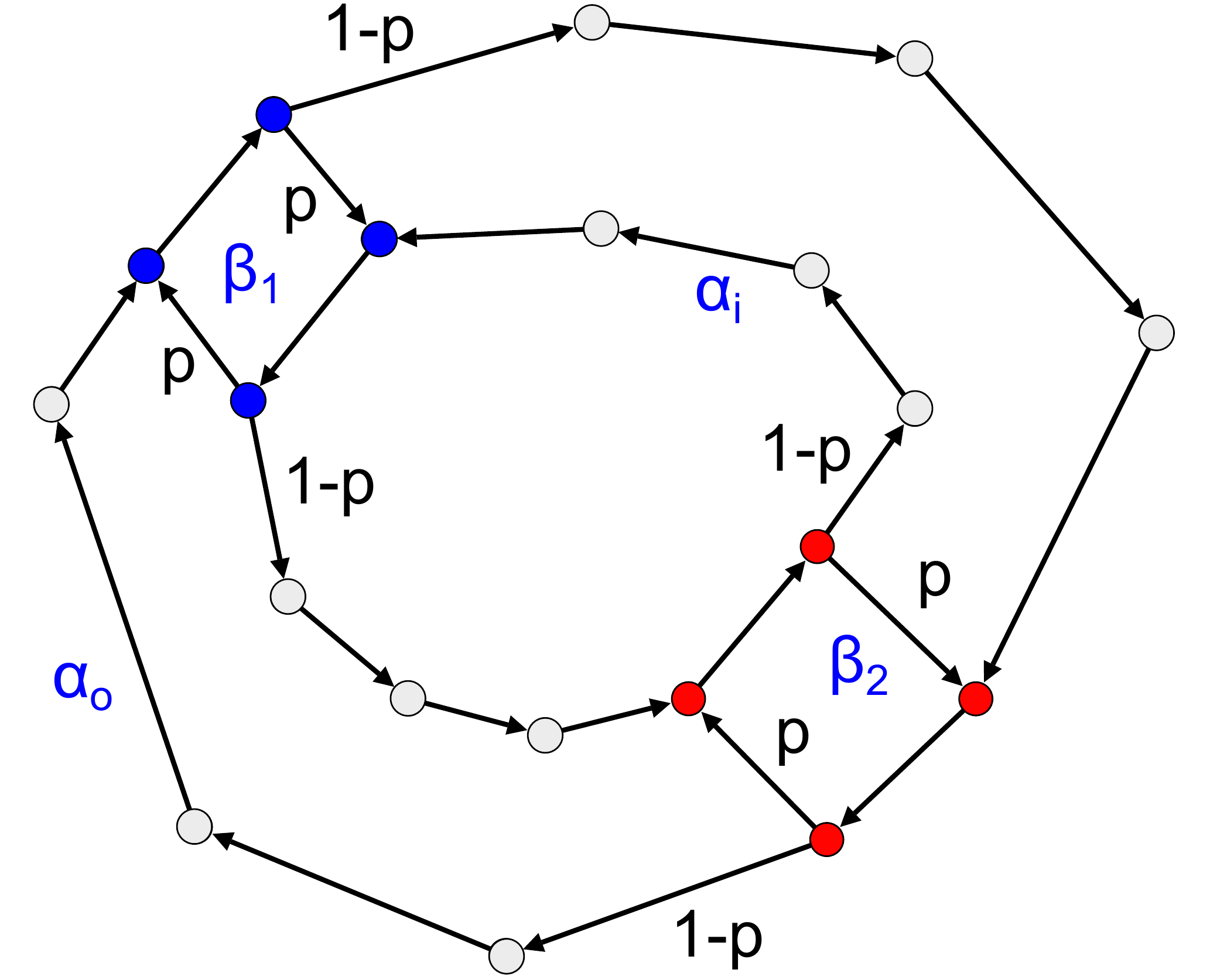}
 \caption{}
 \label{fig:wheel_a}
 \end{subfigure}
 \begin{subfigure}[b]{0.45\textwidth}
 \includegraphics[width=0.95\textwidth]{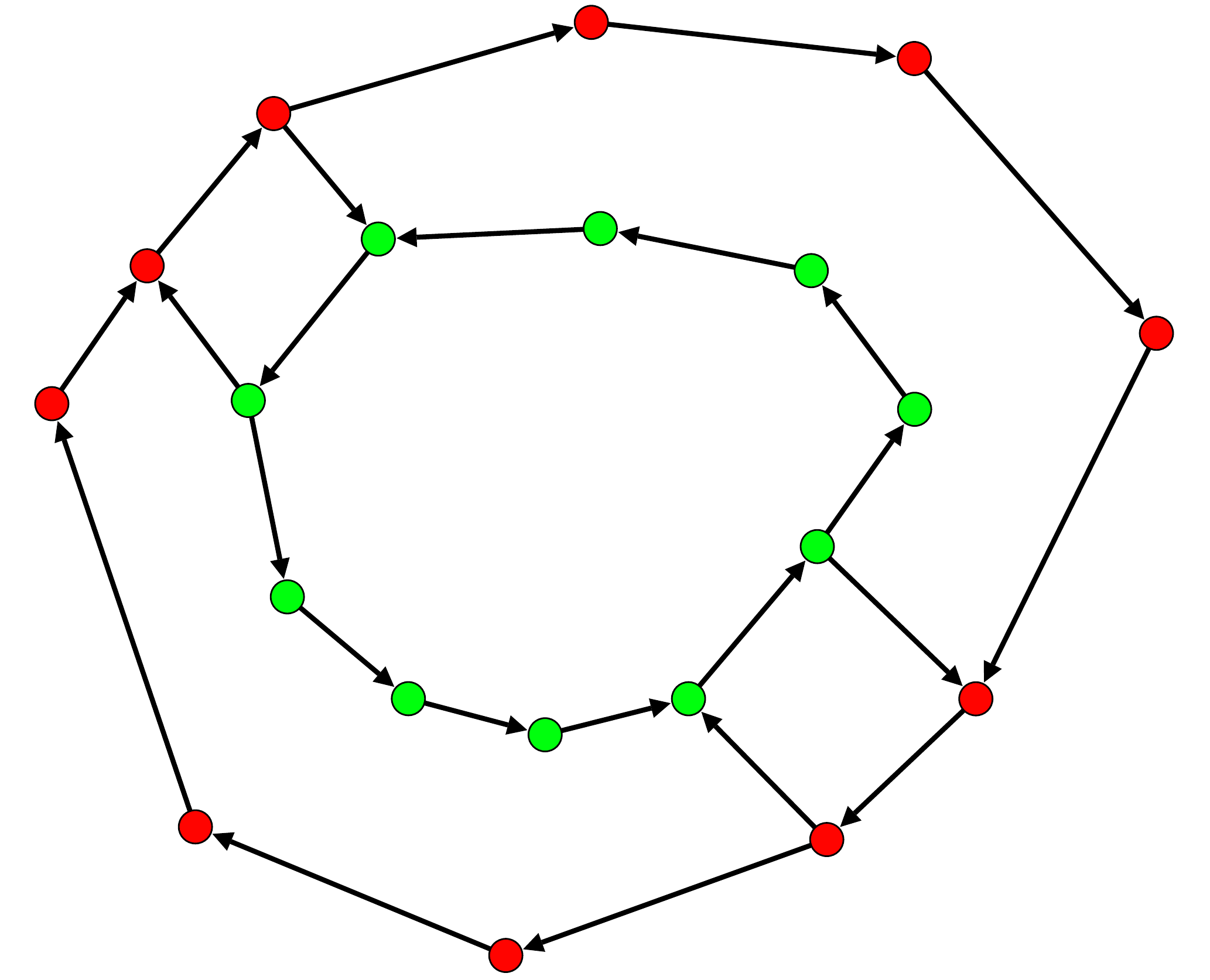}
 \caption{}
 \label{fig:wheel_b}
 \end{subfigure}
 \begin{subfigure}[b]{0.45\textwidth}
   \includegraphics[width=0.95\textwidth]{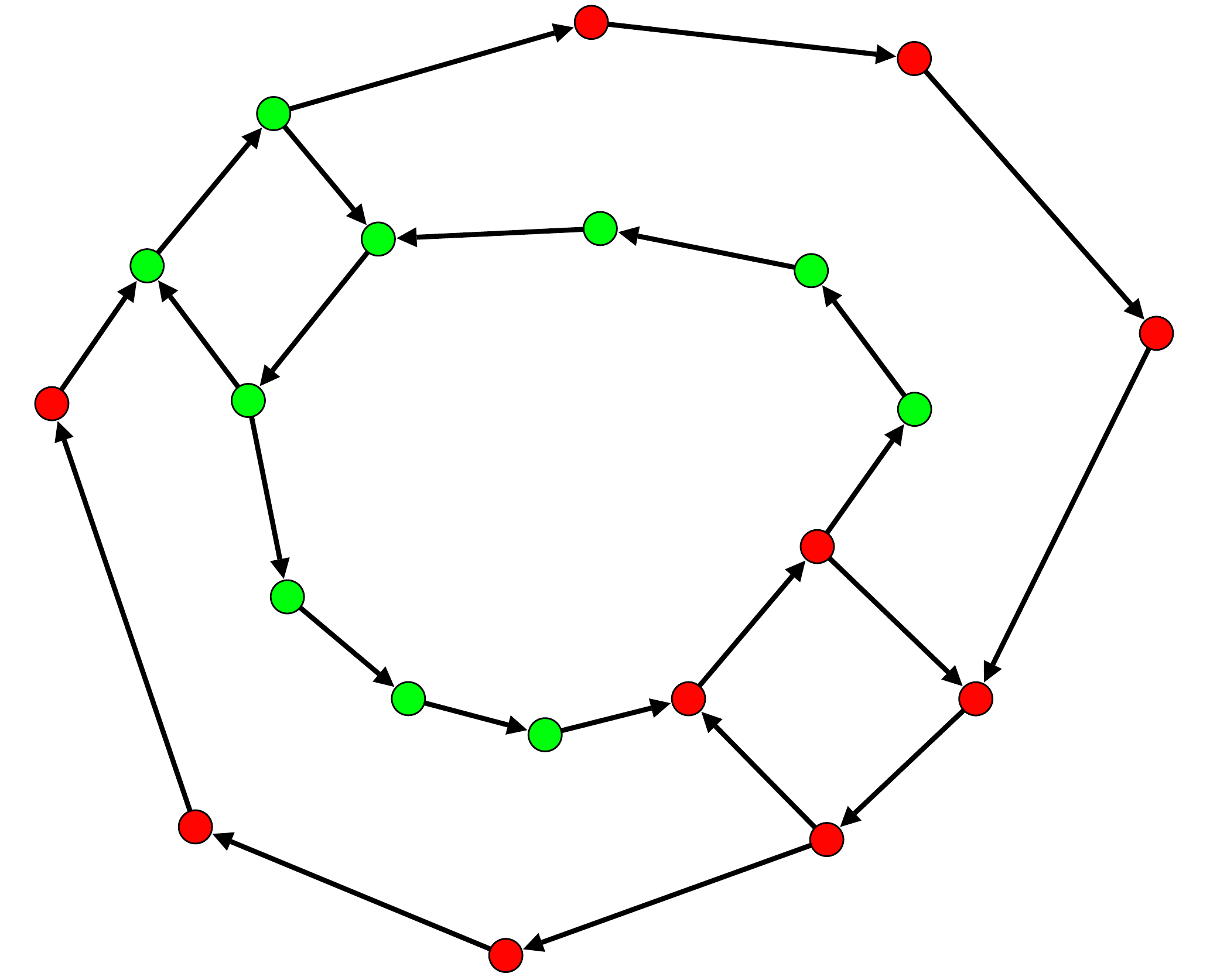}
   \caption{}
   \label{fig:wheel_c}
   \end{subfigure}
 \begin{subfigure}[b]{0.45\textwidth}
  \includegraphics[width=0.95\textwidth]{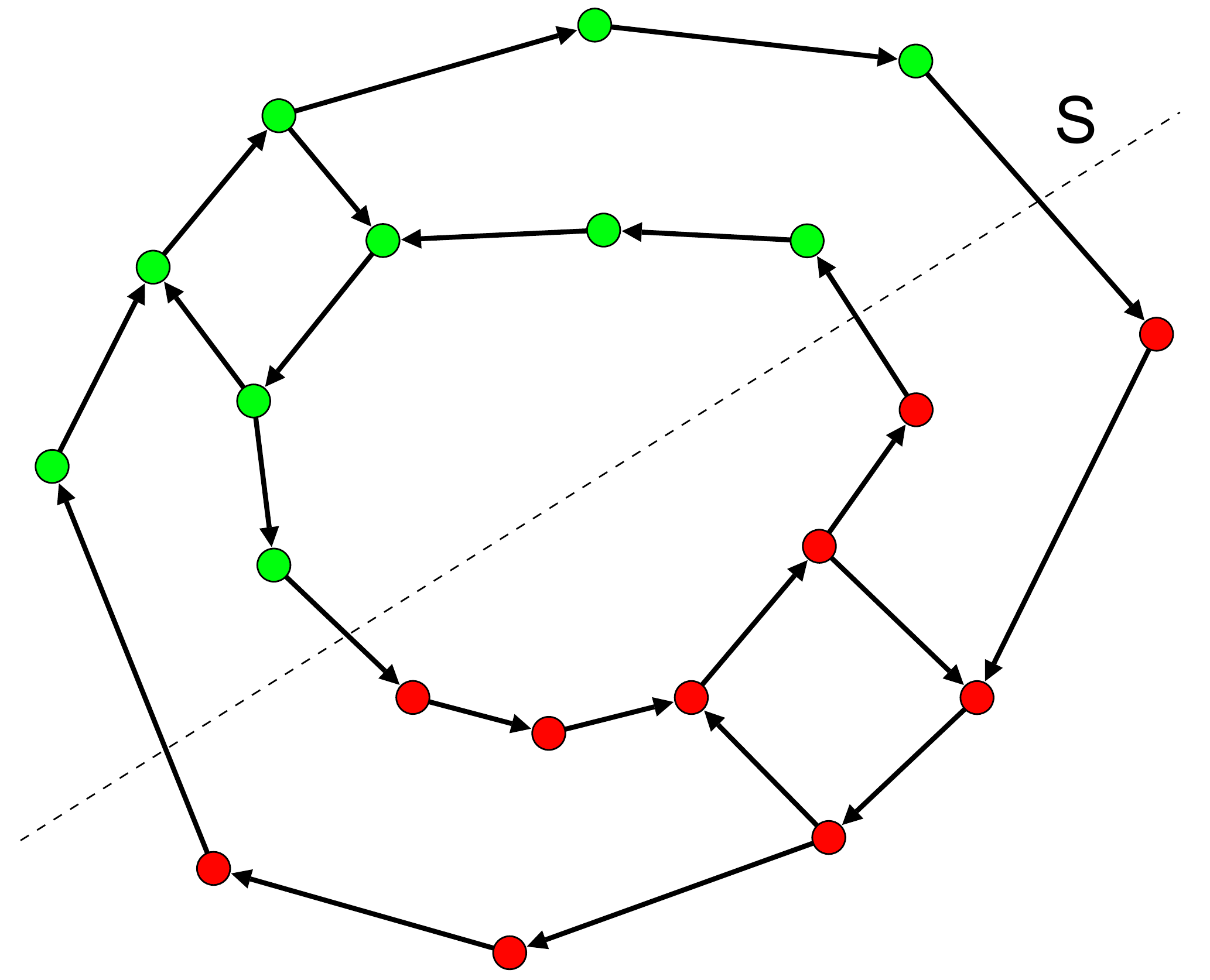}
  \caption{}
  \label{fig:wheel_d}
  \end{subfigure}
  \caption{Graph of the wheel switch with $n=10$, and $p$ being the probability to switch between the outer and inner loop. The coloring: in (\subref{fig:wheel_a}) shows the modules found by CMSM clustering for $p>1/2$; in (\subref{fig:wheel_b}) is shown clustering found by CMSM for $p\leq 1/2$; (\subref{fig:wheel_c}) shows resulting clustering of $\bar Q$-maximization for $p\geq 0.63$; and (\subref{fig:wheel_d}) one example of clustering obtained by $Q$-maximization for $p\geq 1/2$.}
  \label{fig:wheel}
\end{figure}

\section{Conclusion}
In this paper we discussed the problem of finding fuzzy partitions of weighted, directed networks. The main difficulty compared to the well studied case of analyzing undirected networks is that considering asymmetric edge directions often leads to more complicated network substructures. For this reason most of the clustering approaches for undirected networks cannot be generalized to directed networks in a straightforward way. Although different new methods have been proposed, most of them fail to capture all directional information important for module finding, as they consider only one-step, one-directional transitions, which can lead to detecting false modules (where nodes are closely connected only in one direction) and over-partitioning of the network. To this end, we introduced herein a novel measure of communication between nodes which is based on using multi-step, bidirectional transitions encoded by a cycle decomposition of the probability flow. We have shown that this measure is symmetric and that
it
reflects the information transport in the network. Furthermore, this enabled us to construct an undirected communication graph that captures information flow of the original graph and apply clustering methods designed for undirected graphs. We applied our new method to several examples, showing how it overcomes the essential limitation of common methods. This indicates that our method can be used to circumvent the problems of one-step methods, like modularity based methods. Further generalizations of modularity function and consideration of different null models will be the topic of our future research. An important aspect of our method is its role in describing non-equilibrium Markov processes and in particular the connection with entropy production rate. Our future research will be oriented towards these problems.

\section*{Acknowledgments} The authors would like to thank Peter Koltai and Marco Sarich for insightful discussions. RB thanks the Dahlem Research School (DRS) and the Berlin Mathematical School (BMS).

\section*{Appendix}

\subsection*{Proof of equation \ref{eq:Ptilde}:}
We want to find the transition matrix for $(\xi_n)_n$:
\begin{eqnarray*}
 \pP(\xi_{n} = \beta| \xi_{n-1} = \alpha) & = & \sum_x \pP(\xi_n = \beta|\xi_{n-1} = \alpha,  X_n = x)\pP( X_n = x|\xi_{n-1} = \alpha)\\
 & = & \sum_x \pP(\xi_n = \beta| X_n = x)\pP(X_n = x|\xi_{n-1} = \alpha) \\
 & = & \sum_{x\in(\alpha \cap \beta)} b_\beta^{(x)}\pP(X_n = x|\xi_{n-1} = \alpha)
\end{eqnarray*}

by properties of $(X_n)_n, (\xi_n)_n$ and the definition of $b_\beta^{(x)}$. Now

\begin{eqnarray*}
 \pP(X_n=x|\xi_{n-1} = \alpha) & = & \sum_y \pP(X_n = x|\xi_{n-1} = \alpha,X_{n-1} = y)\pP(X_{n-1} = y|\xi_{n-1} = \alpha) \\
 & = & \sum_y \pP(X_n = x|\xi_{n-1} = \alpha,X_{n-1} = y)\\
 & & \times\: \pP(\xi_{n-1} = \alpha|X_{n-1}= y)\frac{\pP(X_{n-1} = y)}{\pP(\xi_{n-1} = \alpha)}\\
 & = & \sum_{y\subset\alpha} \pP(X_n = x|\xi_{n-1} = \alpha,X_{n-1} = y)b_\alpha^{(y)}\frac{\pi_y}{|\alpha|F(\alpha)} \\
 & = & \sum_{y\subset\alpha} \pP(X_n = x|\xi_{n-1} = \alpha,X_{n-1} = y)\frac{1}{|\alpha|}
\end{eqnarray*}

Since $(X_n)_n$ is actually determined by the pair $(\xi_{n-1},X_{n-1})$, we have

$$\pP(X_n = x|\xi_{n-1} = \alpha,X_{n-1} = y) = \begin{cases} 1 & y=n^-_\alpha(x) \\ 0 & \mbox{otherwise}\end{cases}$$

where $n^-_\alpha(x)$ is the vertex we reach next if we follow the cycle $\alpha$ in reverse orientation starting at $x$. Putting it all together, we have

$$\mathcal{Q}_{\alpha\beta} := \pP(\xi_{n} = \beta| \xi_{n-1} = \alpha) = \sum_{x\in (\alpha\cap\beta)} \frac{1}{|\alpha|} b_\beta^{(x)}.$$

\subsection*{Proof of Lemma \ref{lemma:PQprop}:}
(i) follows from (ii) by the fact that $\mathcal{P}$ and $\mathcal{Q}$ are stochastic. To show (ii), $\pi_x \mathcal{P}_{xy} = \pi_y\mathcal{P}_{yx}$ and $\mu(\alpha)\mathcal{Q}_{\alpha\beta} = \mu(\beta)\mathcal{Q}_{\beta\alpha}$ can be directly inferred from (\ref{eq:P}) and (\ref{eq:Q}). (iii) follows from (iv). To show (iv), let $\lambda \neq 0$ and $v\in \ker(\mathcal{Q}-\lambda 1)$. This implies $\mathcal{V}\mathcal{B}v = \lambda v$, multiplication from the left with $\mathcal{B}$ yields $\mathcal{P}\mathcal{B}v = \lambda \mathcal{B}v$, hence $\mathcal{B}v\in \ker(\mathcal{P}-\lambda 1)$. Now since $\mathcal{Q} = \mathcal{V}\mathcal{B}$ we have $\ker \mathcal{B} \subset \ker\mathcal{Q} \perp \ker(\mathcal{Q}-\lambda 1)$. Therefore $\mathcal{B}$ is injective as a map from $\ker(\mathcal{Q} - \lambda 1)$ to $\ker (\mathcal{P} - \lambda 1)$. To show that $\mathcal{B}$ is also surjective, take $w\in \ker(\mathcal{P}-\lambda 1)$. Define $v:=\frac{1}{\lambda}\mathcal{V}w$. Then $v\in \ker(\mathcal{Q}-\lambda
1)$ since $\mathcal{Q}v = \frac{1}{\lambda}\mathcal{Q}\mathcal{V}w = \frac{1}{\lambda}\mathcal{V}\mathcal{P}w = \frac{1}{\lambda}\mathcal{V}\lambda w = \lambda v$, and $\mathcal{B}v = \frac{1}{\lambda}\mathcal{B}\mathcal{V}w = w$.

\bibliographystyle{plain}
\bibliography{jcdpaper}
%
%
%
\end{document}